\documentclass[11pt]{article}
\usepackage{fullpage}
\usepackage{color,graphicx,tikz}
\usepackage{amssymb,amsmath,amsthm}
\usepackage{bm}
\usepackage[italicComments=true,beginComment=/$\!$/~]{algpseudocodex}
\usepackage[normalem]{ulem}
\usepackage[colorlinks, citecolor=blue, linkcolor=black]{hyperref}

\newcommand{\Img}{\mathrm{{Im}}\,}
\newcommand{\ord}{\mathrm{rank}}
\newcommand{\set}[2][]{{\left\{#2\right\}}^{#1}}
\newcommand{\cupdot}{\mathrel{\dot{\cup}}}
\newcommand{\bools}[1]{\set[#1]{0,1}}
\newcommand{\ceil}[1]{\left\lceil{#1}\right\rceil}
\newcommand{\floor}[1]{\left\lfloor{#1}\right\rfloor}
\newcommand{\abs}[1]{\left|{#1}\right|}
\newcommand{\size}[1]{\abs{#1}}

\newcommand{\condset}[2]{\set{#1 \; \left| \;#2 \right. }}

\newcommand{\restrict}[1]{| _{#1}}

\newcommand{\inveps}{\frac{1}{\varepsilon}}
\newcommand\eps{\varepsilon}
\newcommand{\poly}{\mathrm{poly}}
\newcommand{\ct}[1]{\#_{#1}}
\newcommand{\less}{\mathrel{<\hspace{-5pt}\raisebox{1.5pt}{\scalebox{0.4}{E}}}}
\newcommand{\back}{\textsc{Back}}
\newcommand{\class}[1]{\ensuremath{\mathbf{{#1}}}}
\renewcommand{\P}{\class{P}}
\newcommand{\E}{\class{E}}
\newcommand{\EXP}{\class{EXP}}
\newcommand{\FP}{\class{FP}}
\newcommand{\OLtP}{\class{OL_2^{\class{\scriptscriptstyle P}}}}
\newcommand{\prOLtP}{\class{prOL_2^{\class{\scriptscriptstyle P}}}}
\newcommand{\LtP}{\class{L_2^{\class{\scriptscriptstyle P}}}}
\newcommand{\prLtP}{\class{prL_2^{\class{\scriptscriptstyle P}}}}
\newcommand{\SigmatP}{\class{\Sigma_2^{\class{\scriptscriptstyle P}}}}
\newcommand{\StP}{\class{S_2^{\class{\scriptscriptstyle P}}}}
\newcommand{\OtP}{\class{O_2^{\class{\scriptscriptstyle P}}}}
\newcommand{\NOtP}{\class{NO_2^{\class{\scriptscriptstyle P}}}}
\newcommand{\StE}{\class{S_2^{\class{\scriptscriptstyle E}}}}
\newcommand{\OtE}{\class{O_2^{\class{\scriptscriptstyle E}}}}
\newcommand{\LtE}{\class{L_2^{\class{\scriptscriptstyle E}}}}
\newcommand{\NP}{\class{NP}}
\newcommand{\ONP}{\class{ONP}}

\newcommand{\RP}{\class{RP}}
\newcommand{\PH}{\class{PH}}
\newcommand{\Ppoly}{\class{P/poly}}
\newcommand{\MA}{\class{MA}}
\newcommand{\ZPP}{\class{ZPP}}
\newcommand{\BPP}{\class{BPP}}
\newcommand{\BPPp}{\class{BPP_{path}}}
\newcommand{\SBP}{\class{SBP}}
\newcommand{\AM}{\class{AM}}

\newcommand{\prNP}{\class{prNP}}
\newcommand{\prONP}{\class{prONP}}
\newcommand{\prOMA}{\class{prOMA}}
\newcommand{\prOtP}{\class{prO_2^{\class{\scriptscriptstyle P}}}}

\newcommand{\prStP}{\class{prS_2^{\class{\scriptscriptstyle P}}}}
\newcommand{\prMA}{\class{prMA}}
\newcommand{\OprMA}{\class{OprMA}}
\newcommand{\prSBP}{\class{prSBP}}
\newcommand{\prAM}{\class{prAM}}
\newcommand{\prBPPp}{\class{prBPP_{path}}}
\newtheorem{THEOREM}{Theorem}
\newtheorem{COROLLARY}[THEOREM]{Corollary}

\newtheorem{theorem}{Theorem}[section]

\newtheorem{conjecture}[theorem]{Conjecture}
\newtheorem{corollary}[theorem]{Corollary}

\newtheorem{definition}[theorem]{Definition}

\newtheorem{lemma}[theorem]{Lemma}

\newtheorem{observation}[theorem]{Observation}
\newtheorem{remark}[theorem]{Remark}
\newtheorem{proposition}[theorem]{Proposition}
\DeclareFontShapeChangeRule{it}{sc}{sc}{sc}
\newcommand{\problem}[1]{\ensuremath{\textsc{#1}}}
\newcommand{\SAT}{\problem{SAT}}

\newcommand{\ehardtt}{\problem{$\eps$\text{-}Hard-tt}}
\newcommand{\Avoid}{\problem{Avoid}}
\newcommand{\PRG}{\problem{PRG}}
\newcommand{\LOP}{\problem{LOP}}
\newcommand{\MIN}{\problem{MIN}}
\newcommand{\SSE}{\problem{SSE}}
\newcommand{\AP}{\problem{ApproxCount}}
\newcommand{\WSSE}{\problem{WSSE}}

\newcommand{\CircuitSAT}{\problem{CircuitSAT}}
\newcommand{\YES}{\problem{\tiny YES}}
\newcommand{\NO}{\problem{\tiny NO}}

\newcommand{\bsize}{\mathsf{Size}}
\newcommand{\PTIME}{\P}
\newcommand{\eqdef}{:=}
\newcommand{\N}{\mathbb{N}}
\title{Upper and Lower Bounds for the Linear Ordering Principle%
}
\author{%
Edward A. Hirsch\thanks{Department of Computer Science, Ariel University, Israel. This research was conducted with the support of the State of Israel, the Ministry of Immigrant Absorption, and the Center for the Absorption of Scientists. Email: {\tt edwardh@ariel.ac.il}} \and Ilya Volkovich\thanks{Computer Science Department, Boston College, Chestnut Hill, MA. 
    Email: {\tt ilya.volkovich@bc.edu}}%
}
\begin{document}
\maketitle

\begin{abstract}
     Korten and Pitassi (FOCS, 2024) defined a new\footnote{\label{foo:l2p}Note that this notation had been used in the past \cite{Sch83} for a very different class, which has been apparently forgotten after that.} complexity class $\LtP$ as the polynomial-time Turing closure of the Linear Ordering Principle (a total function extending finding the minimum of an order \cite{CK98} to the case where the order is not linear). They put it between $\MA$ (Merlin--Arthur protocols) and $\StP$ (the second symmetric level of the polynomial hierarchy).
     
     In this paper we sandwich $\LtP$ between $\P^\prMA$ and $\P^\prSBP$.
     (The oracles here are promise problems, and $\SBP$ is the only known class between $\MA$ and $\AM$.)
     The containment in $\P^\prSBP$ is proved via an iterative process
     that uses a $\prSBP$ oracle to estimate the average order rank of a subset and find the minimum of a linear order.

     Another containment result of this paper is $\P^\prOtP\subseteq \OtP$
     (where $\OtP$ is the input-oblivious version of $\StP$).
     These containment results altogether have several byproducts:
     \begin{itemize}
         \item We give an affirmative answer to an open question posed by Chakaravarthy and Roy (Computational Complexity, 2011)
        whether $\P^{\prMA}\subseteq \StP$, thereby settling the relative standing of the existing (non-oblivious) Karp–Lipton–style collapse results of \cite{CR11} and  \cite{Cai2007},
        
        \item We give an affirmative answer to an open question of Korten and Pitassi
        whether a Karp--Lipton--style collapse
        can be proven for $\LtP$,
        
        \item We show that the Karp--Lipton--style collapse to $\P^\prOMA$ 
        is actually better than both known collapses to $\P^\prMA$
        due to Chakaravarthy and Roy (Computational Complexity, 2011) and to $\OtP$
        also due to Chakaravarthy and Roy (STACS, 2006).
        Thus we resolve the controversy between previously incomparable Karp--Lipton collapses
        stemming from these two lines of research.
     \end{itemize}

\pagebreak

\tableofcontents

\pagebreak

    \end{abstract}

\section{Introduction}
The seminal theorem of Richard M. Karp and Richard J. Lipton \cite{KarpLipton80} connected non-uniform and uniform complexity by demonstrating a collapse of the Polynomial Hierarchy assuming $\NP$ has polynomial-size Boolean circuits. This collapse has since been very instrumental in transferring lower bounds against Boolean circuits of \emph{fixed-polynomial}\footnote{That is, for any $k \in \N$, the class contains a language that cannot be computed by Boolean circuits of size $n^k$, i.e. a language outside of $\bsize[n^k]$.} size to smaller classes of the Polynomial Hierarchy. Since then, these results were strengthened in many ways
leading to ``minimial''
complexity classes that have such lower bounds and to which the Polynomial Hierarchy collapses.

\subsection{Background}

\subsubsection{Classes Based on Symmetric Alternation}

An important notion in this context is that of \emph{symmetric alternation}. Namely, one of the best collapses was based on the following idea (\cite{Cai2007}, attributed to Sengupta): if polynomial-size circuits for $\SAT$ exist, two provers (defending the answers `yes' and `no', respectively) send such circuits to a polynomial-time bounded verifier who can, in turn, use them to verify membership in any language in $\PH$. The corresponding class $\StP$ \cite{Canetti96,RS98} was thus shown to have fixed-polynomial circuit lower bounds. \textit{(In Section~\ref{sec:defs} we provide formal definitions
for all less known classes we use.)}

Indeed, since $\NP \subseteq \StP$, if $\SAT$ requires superpolynomial circuits, we are done. Otherwise, the Polynomial Hierarchy, which is known to contain ``hard'' languages (that is, for every $k \in \N$, $\PH \not \subseteq \bsize[n^k]$) by Kannan's theorem \cite{Kannan82}, collapses to $\StP$ and so do these hard languages. This technique has been known as a \emph{win-win argument} in the literature \cite{Kannan82,BCGKT96,KW98,Vinodchandran05, Cai2007,Santhanam09,CR11,Volkovich14c,IKV23a}. 
Chen et al. \cite{CMMW19} prove that there is a bidirectional relationship between fixed-polynomial lower bounds and Karp--Lipton--style theorems. 
In the linear-exponential regime,
while the win-win argument can be extended to obtain \emph{superpolynomial} lower bounds for $\StE$ (the linear-exponential version of $\StP$), it falls short of achieving truly \emph{exponential} lower bounds, as it encounters the so-called “half-exponential” barrier (see \cite{MVW99}).

Upon further inspection, one can observe that the presumed polynomial-size circuits for $\SAT$  do not actually depend on the input \emph{itself}, but rather on its \emph{length}. Based on this observation, the collapse was deepened to the \emph{input-oblivious} version of $\StP$, called $\OtP$ \cite{CR06}. Yet, since $\OtP$ is not known and, in fact, \emph{not believed\/} to contain $\NP$, the fixed-polynomial lower bounds
do not (immediately) carry over to $\OtP$. 

This state of affairs remained unchanged for about fifteen years until a significant progress was made when Kleinberg et al. \cite{KKMP21} initiated the study of total functions beyond $\class{TFNP}$. While Karp--Lipton's theorem has not been improved, lower bounds against $\bsize[n^k]$ were pushed down to $\OtP$ \cite{GLV24} and $\LtP$ \cite{KP24}, a new$^\text{\ref{foo:l2p}}$ important class which we describe in more detail below. At the same time, truly exponential lower bounds were established for $\StE$ \cite{li2023symmetric} (as it turns out, $\StE=\OtE$ \cite{GLV24}) and $\LtE$ \cite{KP24}.

An important feature of these new results was that they were based on reducing finding a hard function to a total search problem. Namely, the works of Korten \cite{korten2022hardest} and Li \cite{li2023symmetric} reduced the question to the so-called \emph{Range Avoidance} problem: given a function $f\colon\{0,1\}^n\to\{0,1\}^m$ with $m>n$, represented by a Boolean circuit, find a point outside its image. This problem is known in the bounded arithmetic community as the \emph{dual Weak Pigeonhole Principle} (d\problem{WPHP})
and has notable implications in proof complexity (see e.g. \cite{PWW88,JerdWPHP} and the survey \cite{KraGen25}). In \cite{chen2023symmetric,li2023symmetric}, Range Avoidance has been reduced to symmetric alternation. Subsequently, Korten and Pitassi \cite{KP24} reduced Range Avoidance to the \emph{Linear Ordering Principle}: given an implicitly described ordering relation, either find the smallest element or report a breach of the linear order axioms (for the case of a linear order this problem is known as \problem{MIN}
in the bounded arithmetic community \cite{CK98}).
A polynomial-time Turing closure of this principle gives rise to a new class $\LtP\subseteq\StP$: a version of $\StP$ where 
the two provers provide points of a polynomial-time verifiable linear order on binary strings of a certain length
(each point starting with the corresponding answer $0$ or $1$),
and the prover that provides the smaller element wins.

\subsubsection{Classes Based on Merlin-Arthur Protocols}

In a parallel line of research, the same questions were considered
for classes based on Merlin--Arthur proofs: Santhanam \cite{Santhanam09}
has shown fixed-polynomial lower bounds for \emph{promise problems} possessing such proofs
(i.e. the class $\prMA$).
In \cite{CR11}, Chakaravarthy and Roy have shown a Karp--Lipton--style collapse and thus fixed-polynomial size lower bounds for the class $\P^\prMA$.
In particular, they presented a new upper bound for $\StP$ by showing that $\StP\subseteq\P^\prAM$.
Nonetheless, the relationship between $\P^\prMA$ and the classes of 
symmetric alternation (including $\StP$, $\OtP$, and then-unknown $\LtP$) remained open.

Combining their upper bound for $\StP$ with a result of \cite{AKSS95}, that
$\NP\subseteq\Ppoly$ implies an ``internal collapse'' $\MA=\AM$\footnote{The internal collapse goes through for the promise versions of the classes as well.},  \cite{CR11} concluded that the Polynomial Hierarchy collapses all the way to $\P^\prMA$. 
Subsequently, by applying the win-win argument, they obtained fixed-polynomial bounds for $\P^\prMA$, which (unlike $\prMA$) is a class of languages. It is to be noted though that 
since $\prMA$ is not a class of languages --- while $\P^\prMA$ is, there is no {\itshape immediate\/} way to carry any lower bound against $\prMA$ over to $\P^\prMA$: it is not clear how to leverage (even) Turing reductions to construct a specific language consistent with a given promise problem.

Babai, Fortnow, and Lund \cite{BFL91} prove that if $\EXP\subseteq\Ppoly$,
then $\EXP=\MA$. Although this is a much larger class, the proof has the advantage
that it does not relativize. More collapses in the exponential regime have been proved since then \cite{IKW02, BuhrmanHomer92}, and the win-win argument yields superpolynomial lower bounds for some of them: $\class{MAEXP} \nsubseteq \Ppoly$ \cite{BFT98}.

\subsection{Promise Problems as Oracles}
\label{sec:promise}

An important note is due on the use of a promise problem
as an oracle, because the literature contains several
different notions for this.
The collapse result of Chakaravarthy and Roy that we use follows
the \emph{loose oracle access}
mode adopted in \cite{CR11}. 
Namely, oracle 
queries outside of the promise set are allowed
and no particular behaviour of the computational model
defining the promise class is expected on such queries.
At the same time, the answer of the (base) machine using this oracle
must be correct \textit{irrespective\/} of the oracle's answers
to such queries and no assumption is made on
the internal consistency of the answers outside of the promise set.

For deterministic polynomial-time oracle machines
this approach is equivalent to querying any language
\emph{consistent} with the promise problem, that is,
a language that contains all the ``yes'' instances and does not contain the ``no'' instances of the promise problem.
One can also extend it to complexity classes: 
a class $\mathcal{C}$ of languages is consistent with 
a class $\mathcal{D}$ of promise problems
if for every problem $\Pi\in\mathcal{D}$ there is a language $L\in\mathcal{C}$
consistent with $\Pi$.
The equivalence between the approaches follows from the works of \cite{GS88,BF99},
where the latter approach has been adopted. Nonetheless, we include a formal proof of this equivalence in Section~\ref{subsec:loose} for the sake of completeness
and for further reference.

\subsection{Our Contribution}

In this paper we prove the inclusions $\P^\prMA\subseteq \LtP\subseteq\P^\prSBP$ and $\P^\prOtP\subseteq \OtP$,
which not only give new upper and lower bounds for $\LtP$,
but also demonstrate that the Karp--Lipton--collapse
to $\P^\prOMA$ is currently the best one both for symmetric--alternation--based and Merlin--Arthur--based classes of languages.

\subsubsection{A New Lower Bound for $\LtP$ and the Strongest Non-Input-Oblivious Karp--Lipton Collapse}

Two open questions regarding symmetric alternation have been stated explicitly:
\begin{itemize}
    \item whether $\P^\prMA$ is contained in $\StP$ \cite{CR11} (note that for these two classes
    Karp--Lipton style theorems have been proved by \cite{CR11} and by Samik Sengupta [unpublished], respectively),
    \item whether a Karp--Lipton--style theorem holds for $\LtP$ \cite{KP24}.
\end{itemize}
In this paper we resolve both these questions affirmatively by showing the following containment. (Recall that $\LtP \subseteq \StP$.)

\begin{THEOREM}
\label{THM:main1}
$\P^\prMA\subseteq \LtP$.    
\end{THEOREM}
Combining this theorem with a result of Chakaravarthy and Roy \cite{CR11} that $\NP\subseteq\Ppoly$ implies the collapse $\PH=\P^{\prMA}$, we obtain a Karp--Lipton--style collapse theorem for $\LtP$, thus resolving an open question posed in \cite{KP24}.

\begin{COROLLARY}
\label{COR:main2}
    If\/ $\NP\subseteq\Ppoly$,
    then\/ $\PH=\LtP=\P^{\prMA}$.
\end{COROLLARY}

Together with the result of \cite{Cai2007}, it
lines up all known non-input-oblivious 
classes for which a Karp--Lipton--style collapse has been shown:
\[\P^\prMA\subseteq\LtP\subseteq\StP\subseteq\ZPP^\NP\subseteq\SigmatP.\]
\subsubsection{A New Upper Bound for $\LtP$}
Another important result of this work is a new upper bound on $\LtP$:
we prove that $\LtP\subseteq\P^\prSBP$.
The best known upper bound prior to our result followed from \cite{KP24,CR11}:
$\LtP\subseteq\StP\subseteq\P^\prAM$.

\begin{THEOREM}
\label{THM:main3}
$\LtP\subseteq\P^\prSBP$.
\end{THEOREM}

\noindent In summary, our two new inclusions (Theorems~\ref{THM:main1} and~\ref{THM:main3}) yield the ``normal'' (non-input-oblivious) part of 
Figure~\ref{fig:main}.

\begin{figure}
\begin{center}
\begin{tikzpicture}[scale=1.5]
\node(MA) at (-3,0) {$\MA$};    
\node(SBP) at (3.3,0) {$\SBP$};    
\node(AM) at (5,0) {$\AM$};    
\node(PprMA) at (-2.5,1) {$\P^\prMA$};    
\node(LtP) at (1.5,1) {$\LtP$};
\node(PprSBP) at (3.8,1) {$\P^\prSBP$};    
\node(StP) at (3,2) {$\StP$};    
\node(PprAM) at (5.5,1) {$\P^\prAM$};    
\draw[->](MA) to (SBP);
\draw[->](SBP) to (AM);
\draw[->](PprMA) to node[above]{\textcolor{blue}{\footnotesize Theorem~\ref{THM:main1}}}  (LtP);
\draw[->](LtP) to node[above]{\textcolor{blue}{\phantom{Wi}\footnotesize Theorem~\ref{THM:main3}}}  (PprSBP);
\draw[->] (LtP) to (StP);
\draw[->](PprSBP) to (PprAM);
\draw[->](StP) to (PprAM);
\draw[->](MA) to (PprMA);
\draw[->](SBP) to (PprSBP);
\draw[->](AM) to (PprAM);
\textcolor{green!50!black}{
\footnotesize
\node(PprOMA) at (-3.5,3) {$\P^\prOMA$};    
\node(POLtPprONP) at (-1.5,3) {$\P^{\OLtP,\prONP}_{\vphantom{\OLtP}}$};
\node(OtPLtP) at (0.5,3) {$\OtP\cap\LtP$};
\node(PprOtP) at (2.5,4) {$\P^\prOtP$};
\node(OtP) at (2.5,3) {$\OtP$};
\draw[->](PprOMA) to node[below]{\textcolor{blue}{\scriptsize Cor.$\,$\ref{cor:OLtP}}}  (POLtPprONP);
\draw[->](POLtPprONP) to node[below]{\textcolor{blue}{\scriptsize Cor.$\,$\ref{cor:OLtP}}} (OtPLtP);
\draw[->](POLtPprONP) to (PprOtP);
\draw[double](PprOtP) to node[right] {\small\textcolor{blue}{$\,$\footnotesize Theorem~\ref{THM:main5}}} (OtP);
\draw[->](OtPLtP) to (OtP);
\draw[->](OtPLtP) to (LtP);
\draw[->](PprOMA) to (PprMA);
\draw[->](OtP) to (StP);
}
\end{tikzpicture}
\end{center}
\caption{Containments of classes based on Merlin--Arthur protocols and on symmetric alternation.}\label{fig:main}
\end{figure}

\pagebreak
Switching to the linear-exponential regime,
in \cite{KP24}, Korten and Pitassi have shown that $\LtE$ --- the exponential version of $\LtP$ --- 
contains a language of circuit complexity $2^n/n$.
By translation, our upper bound scales as $\LtE \subseteq \E^\prSBP${\ }\footnote{As is turns out, this class is also equal $\E^\prBPPp$. See Section \ref{sec:int:approx} for more details.}.
As a corollary we obtain a new circuit lower bound for the class $\E^\prSBP$.
To the best of our knowledge, the strongest previously established bound for this class was ``half-exponential''
that followed from the bound on $\class{MAEXP}$ \cite{MVW99}. 

\begin{COROLLARY}
\label{COR:main4}
$\E^\prSBP$ contains a language of circuit complexity $2^n/n$.
\end{COROLLARY}

It is to be noted that Corollary \ref{COR:main4} could be viewed as an \emph{unconditional} version of a result of Aydinlio\~glu et al. \cite{AGHK11} as it recovers and strengthens their conclusion.  
In particular, 
\cite{AGHK11} have shown the following\footnote{The corresponding claims in \cite{AGHK11} were stated using slightly different terminology.}:  
if $\P^\NP$ is consistent with $\prAM$ or even $\prSBP$, then 
$\E^\NP$ contains a language of circuit complexity $2^n/n$.  
Indeed, given the premises we obtain that $\E^\prSBP \subseteq \E^{\P^\NP} =
\E^\NP$ from which the claim follows directly by Corollary \ref{COR:main4}.

\subsubsection{Aggregation of $\prOtP$ queries and the Strongest Input-Oblivious Karp--Lipton Collapse}

Theorem~\ref{THM:main1} shows that $\P^{\prMA}$ is currently the smallest non-input-oblivious class  for which a Karp--Lipton--style collapse is known. On the other hand, such a collapse was also shown for $\OtP$ \cite{CR06}, which is input-oblivious. However, since the precise relationship between $\OtP$ and $\P^\prMA$ remains unknown, one may ask: what is the strongest Karp--Lipton--style collapse? Our next result assists in navigating this question.

\begin{THEOREM}
\label{THM:main5}
$\P^\prOtP \subseteq \OtP$.
\end{THEOREM}

We note that the ``non-promise'' version of this inclusion, i.e. $\P^\OtP \subseteq \OtP$, was already established in \cite{CR06}. However, this result does not carry over to the promise case.
A similar phenomenon arises in the non-input-oblivious analogue of this question: while we know that  $\P^\StP\subseteq\StP$ \cite{CR06}, it still remains open whether $\P^\prStP\subseteq\StP$.

With this tool in hand, we can identify and show the strongest Karp--Lipton--style collapse
that currently known. Namely,
the collapse can be extended to $\P^\prOMA \subseteq \P^\prMA$, where $\prOMA$ is the input-oblivious version of $\prMA$ and therefore is contained in $\prMA$.
At the same time, Theorem~\ref{THM:main5} allows us to show that $\P^\prOMA$ is also contained in $\OtP$,
thus making $\P^\prOMA$ smaller than both $\P^{\prMA}$ and $\OtP$!\\

Theorem~\ref{THM:main5} also allows us to refine the chain of containments between $\P^\prOMA$ and $\OtP$ by introducing a newly defined class, $\OLtP$, an input-oblivious analogue of $\LtP$ (see Figure~\ref{fig:main}). Note that this chain proceeds via $\P^{\ONP,\OLtP}$ rather than $\OLtP$ itself, since $\OLtP$ is not guaranteed either to share the convenient closure properties of $\LtP$ or to contain $\ONP$.
Still we prove that
$$\P^\prOMA\subseteq\P^{\OLtP,\prONP}\subseteq\OtP\cap\LtP.$$
That is, $\P^{\OLtP,\ONP}$ can serve as an input-oblivious analogue of $\LtP$ ($=\P^{\LtP,\NP}$ \cite{KP24}).

\subsection{Our Techniques}
\subsubsection{Approximate Counting, Set Size Estimation and $\SBP$}
\label{sec:int:approx}

Computing the number of accepting paths of a given non-deterministic Turing machine is a fundamental problem captured by the ``counting'' class $\mathbf{\#\P}$. Yet, this class appears to be too powerful since, by Toda's Theorem \cite{Toda91}, even a single query to it suffices to decide any language
in the polynomial hierarchy, $\PH\subseteq\P^{\mathbf{\#\P[1]}}$!
Given that, it is natural to explore approximations.  To this end, one can consider the problem of \emph{Approximate Counting}  ($\AP$ for short) which refers to the task of \emph{approximating} the number of accepting paths
(within a constant factor). 
Equivalently, this problem can be framed as approximating the size of a set $S$ represented as
the set of satisfying assignments of a Boolean circuit $C$. 
Previously, it was shown that this task could be carried out by a randomized algorithm using an $\NP$ oracle ($\class{FBPP}^\NP$) \cite{Stockmeyer85,JVV86} 
and by a deterministic algorithm using a $\prAM$ oracle ($\FP^\prAM$) \cite{Sipser83,GS86}. Shaltiel and Umans \cite{ShaltielU06} show how to accomplish this task in $\FP^\NP$, 
yet under a derandomization assumption.
We note that all of these algorithms can be implemented using parallel (i.e. non-adaptive) oracle queries. That is, in $\class{FBPP}^\NP_{\parallel},\FP^\prAM_{\parallel}$ and $\FP^\NP_{\parallel}$, respectively.
In \cite{Jer2007}, Je\v{r}\'abek studied approximate counting in the context of bounded arithmetic and reduced to it many problems in various complexity classes.
Approximate counting has also recently attracted considerable attention in the quantum literature \cite{OS18, AKKT20, AR20, MAD25}. \\

The decision version of the problem is to distinguish between two constant-factor \emph{estimates} of the set size. For concreteness, consider the following problem called \emph{set-size estimation} (or $\SSE$ for short):
Given a set $S$ (via a Boolean circuit $C$) and an integer $m$ with the \emph{promise} that either $\size{S} \geq m$ or $\size{S} \leq m/2$,  our goal is to decide which case holds.
Interestingly, this problem is complete for the class $\SBP$\footnote{Strictly speaking, the problem is complete for $\prSBP$, the corresponding class of promise problems, and is therefore hard for $\SBP$.} introduced in \cite{BGM06} by B{\"{o}}hler et al.
as a relaxation of the class $\BPP$ to the case when the acceptance probability is not required to be bounded away from $0$.
This relaxation, as was shown in \cite{BGM06}, yields additional power:  $\MA \subseteq \SBP \subseteq \AM$. 
The class $\SBP$, which stands for \emph{small bounded-error polynomial-time}, thus sits strictly between the two fundamental classes based on Arthur--Merlin protocols, yet its definition is not based on these protocols. Moreover, $\SBP$ remains the only known natural class that lies between $\MA$ and $\AM$.

In terms of upper bounds, in a seminal paper $\cite{GS86}$,
Goldwasser and Sipser have exhibited an Arthur--Merlin protocol not only for this problem,
but also for the case when the set $S$ is represented by a \emph{non-deterministic} circuit\footnote{A non-deterministic circuit $C(x,w)$ accepts $x$ if there exists a witness $w$ for which $C(x,w) = 1$.}! 
This more general version of the problem ($\WSSE$, where the set $S$ is given via a non-deterministic circuit), is complete for the class $\prAM$. (See Definition \ref{def:SSE/WSSE} for the formal definition of the problems; in fact, the factor-of-two gap in the estimates is arbitrary and 
can be replaced by any positive constant.) In fact, Goldwasser--Sipser's protocol proves the containment 
 both for languages ($\SBP \subseteq \AM$) and for promise problems ($\prSBP \subseteq \prAM$). At the same time,
it is important to highlight  the distinction between the two versions of the problem --- i.e. for the ``standard'' ($\SSE$) vs non-deterministic ($\WSSE$) Boolean circuits --- which appears to be (at the very least) non-trivial. Notably, the work of \cite{BGM06} established  an oracle separation between $\SBP$ and $\AM$. \\

On a similar note, by combining some of the previous techniques, we observe that $\AP$ can be carried out in $\FP^\prSBP$ rather than $\FP^\prAM$, and, in fact, even in $\FP^\prSBP_{\parallel}$. 
Given this observation, it is natural to study the computational power of $\P^{\AP}$, that is, deterministic algorithms with oracle access to $\AP$.
Indeed, an immediate corollary of the above is that $\P^{\AP} = \P^{\prSBP}$ and $\P^{\AP}_\parallel = \P^{\prSBP}_\parallel$. At the same, O'Donnell and Say \cite{OS18} previously showed that $\P^{\AP}_\parallel = \BPPp$, a complexity class defined earlier by Han et al. \cite{HHT97}.
One can think of $\BPPp$ as a version of $\BPP$ in which different computational paths (of the same probability) may have different lengths.
Incidentally, it was established in  \cite{BGM06} that $\BPPp$  can be obtained from $\BPP$ via the so-called ``posts{\nobreak}election'' and that 
$\SBP \subseteq \BPPp$ (and resp. $\prSBP \subseteq \prBPPp$).  
Putting all together, one arrives at the following three clusters of complexity classes associated with approximate counting:
$$ \colorbox{blue!15!white}{$\SBP\vphantom{\P^{\prSBP}_\parallel}$}  \subseteq \colorbox{green!30!white}{$\BPPp = \P^{\AP}_\parallel = \P^{\prSBP}_\parallel$} \subseteq  
\colorbox{yellow!40!white}{$\vphantom{\P^{\prSBP}_\parallel}\P^{\AP} = \P^{\prSBP} = \P^{\prBPPp}$}, $$

\noindent where both inclusions are believed to be strict.

\subsubsection{Approximate Counting and the Order Rank Approximation}
\label{sec:approx}
The upper bound $\LtP\subseteq\P^\prSBP$ is obtained
by developing a process that, given an arbitrary element
in a linearly ordered set, rapidly converges to the set's minimum.

\paragraph{Approximate counting using a $\prSBP$ oracle.}
We show how to deterministically approximate the number of satisfying assignments of a  Boolean circuit, given oracle access to $\prSBP$ (i.e. in $\FP^\prSBP$), using parallel queries. Our algorithm is based on 
$\SBP$ amplification that was used in \cite{BGM06,Volkovich20}.
A crucial observation is that,
as we need a multiplicative approximation (up to the factor $1+\varepsilon$),
it suffices to place the desired number between two
consecutive powers of two; the correct place then could be found
by either querying a $\prSBP$ oracle $O(n/\varepsilon)$
times in parallel or (using binary search) $O(\log_2(n/\varepsilon))$ times sequentially.
This result could be on independent interest. 
See Lemma \ref{lem:SBP-set-approx} for the formal statement.

\paragraph{Approximating the order rank w.r.t.~a linear order.}
The \emph{order rank} of an element $\alpha$ of a linearly ordered set $U$ is the number of elements in
this set that are strictly less than $\alpha$
(in particular, $\alpha$ is the minimum if and only if $\ord(\alpha) = 0$).
We can extend this definition to non-empty subsets $S \subseteq U$, 
where $\ord(S)$ is the average order rank of  elements in $S$.

We reduce the problem of approximately comparing the average order ranks of two sets to approximate counting.
To see how, consider a strict linear order $\less$ implicitly defined on $U=\{0,1\}^n$
using a Boolean circuit $E$,
and observe that for a non-empty subset $S\subseteq U$,
the average order rank of $S$ is exactly
the size of the set of pairs
$\{(\upsilon,\alpha)\in U\times S\;|\;\upsilon\less\alpha\}$
divided by the size of $S$.
Hence, this task can be carried out using a $\prSBP$ oracle.

\paragraph{An upper bound for the Linear Ordering Principle.}
As was mentioned, we develop a process that, given an arbitrary element in a linearly ordered set $U = \{0,1\}^n$, rapidly converges to the set's minimum. 

Given an element $\alpha \in U$, we first define the set $S$ as the set of all the elements less or equal to $\alpha$. Formally, $S \eqdef \condset{x}{x \leq \alpha}$. Observe that $\ord(S) = \ord(\alpha)/2$. We then iteratively partition $S$ into two disjoint sets, starting from $i=1$:
$$ S_0 = \condset{x \in S}{x_i = 0} \; \text{ and } S_1 = \condset{x \in S}{x_i = 1}.$$
By averaging argument, $\min \set{ \ord(S_0), \ord(S_1) } \leq \ord(S)$. We then take $S$ to be the subset ($S_0$ or $S_1$) with the smaller order rank and continue to the next value of $i$. That is, we fix the bits of the elements of $S$ one coordinate at a time. Therefore, once $i=n$, our ``final'' set $S$ contains \emph{exactly\/} one element $\beta$ and thus at that point $\ord(S) = \ord(\beta)$. On the other hand, as the order rank of the ``initial'' $S$ was $\ord(\alpha)/2$ and the overall order rank could only decrease, we obtain that  $\ord(\beta) \leq \ord(\alpha)/2$.
We can then invoke the same procedure this time with $\beta$ as its input. As the there are $2^n$ elements in $U$, this process will converge to the set's minimum after invoking the procedure at most $n$ times, given \emph{any\/} initial element.

The algorithm described above requires computing (or at least comparing) the average order ranks of two sets. Our analysis demonstrates that a procedure for approximate comparison, developed before, is sufficient for the implementation of this idea (though the factor at each step will be a little bit less than $2$).

\subsubsection{Derandomization in {\LtP}}
In \cite{KP24}, Korten and Pitassi show that $\MA\subseteq\LtP$.
The inclusion $\P^\prMA\subseteq\LtP$ essentially follows their argument with the additional observation that since $\LtP$ is a syntactic class, not only it  \emph{contains}\/ $\MA$ (as was shown) 
but it is also \emph{consistent}\/ with $\prMA$. 
Thus one can first construct a pseudorandom generator using
an $\LtP$ oracle \cite{korten2022hardest,KP24} and then leverage it
to fully derandomize the $\prMA$ oracle not just in $\prNP$, but actually in $\NP\subseteq\LtP$!
Therefore, $\P^\prMA \subseteq \P^\LtP = \LtP$.

We also observe that since the pseudorandom generator
depends only on the input length (and not the input itself), 
derandomization also helps settling the relations
between input-oblivious classes to a certain extent.
The main difference is that unlike their non-oblivious counterparts, they do not posses all the desired properties w.r.t. Turing closure and natural containments. However,
Theorem~\ref{THM:main5} (its technique
is described below) eventually helps us building
the chain from $\P^\prOMA$ to $\OtP$
in two different ways: both directly 
and through derandomization and intermediate classes using $\OLtP$.

\subsubsection{Input-Oblivious Symmetric Alternation}

A Karp--Lipton--style collapse to $\P^\prOMA$ follows from \cite{CR11}
by combining several previously known techniques.
However, is this collapse stronger than the known collapse to $\OtP$ \cite{CR06}?
The inclusion $\prOMA\subseteq\prOtP$ can be transferred
from a somewhat similar statement that was proven in \cite{CR06};
however, in order to prove $\P^\prOMA\subseteq\OtP$
we need also the inclusion $\P^\prOtP\subseteq\OtP$,
which seems novel. The main idea is that the two provers
corresponding to the oracle give their input-oblivious certificates
prior to the whole computation, and the verification algorithm
performs a cross-check not only between the certificates of \textbf{different}
provers but also between the certificates of the \textbf{same} prover,
which allows us to simulate all oracle queries to $\prOtP$ in a single
$\OtP$ algorithm.
Indeed, our approach is made possible by the input-oblivious nature of the computational model: while the 
oracle queries may be adaptive
and not known in advance (due to potential queries outside of the promise set),
the certificates are universal for the whole computation and nothing else is required.



\subsection{Organization of the Paper}

The paper is organized as follows: In Section~\ref{sec:defs} we give the necessary definitions and also discuss oracle access modes in a more formal way.
Section~\ref{sec:PprSBP} contains the proof of Theorem \ref{THM:main3} -- a new upper bound on $\LtP$. Section~\ref{sec:bestKL}
contains the proof of Theorem \ref{THM:main5}
which implies that collapse to $\P^\prOMA$ subsumes both collapses
to $\P^\prMA$ and to $\OtP$. Section~\ref{sec:PprMA} contains the proof of Theorem \ref{THM:main1} answering the open questions of Korten and Pitassi \cite{KP24} and Chakaravarthy and Roy \cite{CR11}. We also include a version of this theorem
for input-oblivious classes.
We discuss
multiple research directions
in Section~\ref{sec:OQ}.

\section{Definitions}\label{sec:defs}
\subsection{Classes of Promise Problems as Oracles}\label{subsec:loose}

Before defining complexity classes, we first clarify what we mean by oracle access to classes of promise problems. 
A \emph{promise problem} is a relaxation of (the decision problem for) a language.

\begin{definition}[promise problem]
$\Pi= (\Pi_{\YES}, \Pi_{\NO})$ is a \emph{promise problem} if\/ $\Pi_{\YES} \cap \Pi_{\NO} = \emptyset$. 
\end{definition}

Similarly to 
\cite{CR11}, 
when an oracle is described as a promise problem,
we use \emph{loose access} to the oracle.
The outer Turing machine is allowed to
    make queries outside of the promise set,
    and the oracle does not need to conform to
    the definition of the promise oracle class for such queries.
    However, the outer Turing machine 
    must return the correct answer
    \textit{irrespective\/} of oracle's behavior for queries
    outside of the promise set
    in particular, the oracle does not need to be consistent
    in its answers to the same query. Let us now formulate it in a more precisely to avoid misunderstanding:
\paragraph{Loose Oracle Access vs Access Through a Consistent Language}
    One must be very careful when arguing about oracle access to promise classes.
    Several modes of access have been considered in the literature including
    guarded access \cite{LR94}, loose access (e.g., \cite{CR11,HS15}), and access through
    a consistent language \cite{GS88,BF99}.

    In this paper we use the same mode of access as in  \cite{CR11}. Fortunately, when the underlying computational model is
    a polynomial-time Turing machine [with a specific polynomial-time alarm clock],
    loose access is equivalent to access through a consistent language.
    In order to make this connection very clear,
    below we define both of them rigorously and show the equivalence.
    Note that for other computational devices, such as those
    corresponding to $\OtP^\bullet$ or $\StP^\bullet$,
    the connection between the modes is much less clear.

\newcommand\subsetsim{\mathrel{%
  \ooalign{\raise0.2ex\hbox{$\subset$}\cr\hidewidth\raise-0.8ex\hbox{\scalebox{0.9}{$\sim$}}\hidewidth\cr}}}

\paragraph{Access through a consistent language.} The following definition is from \cite{GS88}.
\begin{definition}[consistency]
Consider a promise problem
$\Pi= (\Pi_{\YES}, \Pi_{\NO})$, where $\Pi_{\YES} \cap \Pi_{\NO} = \emptyset$. 
We say that a language $O$ is \emph{consistent} with\/ $\Pi$,
if\/ $\Pi_\YES\subseteq O$ and $\Pi_\NO\subseteq\overline{O}$.
Let us denote it as $O\subsetsim\Pi$.
Note that the containment of $O$ outside of $\Pi_{\YES} \cup  \Pi_{\NO}$ can be arbitrary.

\end{definition}

In one line of the literature \cite{GS88,BF99}, 
when an oracle is described as a promise problem,
the following mode of access is used (\emph{access through a consistent language}):
the outer Turing machine queries
a language $L$ consistent with this promise problem.
In particular, it is allowed to
make queries outside of the promise set.
However, the outer Turing machine 
must return a correct answer
for every possible choice of $L$.

Let us formalize it in the simplest case of $\P^\Pi$:
\begin{definition}
    A language $L$ is in $\P^\Pi$ according to the consistent language mode,
    if there is a deterministic oracle Turing machine $M^\bullet$
    with polynomial-time alarm clock stopping it in time $p(n)$
    such that for every $O\subsetsim\Pi$ and for every input $x$,
    $M^O(x)=L(x)$. 
\end{definition}
One can observe that for every such language $O$, $\PTIME^{\Pi} \subseteq \PTIME^{O}$,
therefore
the following holds as well, and in fact can be considered as an alternative definition.
\begin{definition}[promise problems as oracles] 
\label{def:promise ref}
$\PTIME^{\Pi}  
\eqdef  \bigcap \limits _{
\substack{\Pi_\YES\subseteq O\\\Pi_\NO\subseteq\overline{O}}
} \PTIME^O$.  
\end{definition}
For more details and discussion see e.g. \cite{GS88,BF99}. 

\paragraph{Loose access.}
In the loose access mode, one considers oracle as a black box
(physical device) that is only guaranteed to work correctly
on the promise set. Outside this set it can produce any answer,
in particular, it can give different answers to the same query,
either during one deterministic computation,
or for different witnesses if the computational model uses them,
or for different inputs.

Let us define it rigorously in the simplest case of $\P^\Pi$:
\begin{definition}
    A language $L$ is in $\P^\Pi$ according to the loose access mode,
    if there is a deterministic oracle Turing machine $M^\bullet$
    with polynomial-time alarm clock stopping it in time $p(n)$
    such that for every input $x$
    $M^\bullet(x)$ produces the answer $L(x)$ 
    if the oracle answers ``yes'' on queries in $\Pi_Y$ and
    answers ``no'' on queries in $\Pi_N$.
    (No hypothesis is made on its behaviour outside
    $\Pi_Y\cup\Pi_N$, in particular, its answers
    may be inconsistent for the same query.)
\end{definition}

\paragraph{Equivalence for $\P^\Pi$.}
Fortunately, it is very easy to verify the equivalence of these two
definitions in the case of $\P^\Pi$.
\begin{proposition}
    Let $\Pi$ be a promise problem.
    A language $L$ belongs to $\P^\Pi$ according to the loose access mode
    if and only if it belongs to $\P^\Pi$ according to access through a consistent language.
\end{proposition}
\begin{proof}
    The ``only if'' direction is trivial since a consistent language
    is a particular case of a black box.
    For the ``if'' direction, consider $M^\bullet$ from the definition
    of access through a consistent language.
    To make it tolerant to the loose access mode,
    maintain a table of oracle answers,
    and if $M^\bullet$ attempts to repeat
    the same query, use the answer from the table instead.
    The definition of the machine will not change if
    its oracle is a language, thus it will still produce
    the same answer. On the other hand, 
    its behaviour on input $x$ when given a \emph{blackbox} oracle
    is the same as its behaviour when
    given the following \emph{language} $L_x$ as an oracle:
    $L_x=\Pi_Y\cup Y_x$, where $Y_x$ is the set of
    all queries for which the blackbox oracle
    gave the answer ``yes'' when asked for the first time.
    Since $L_x$ is consistent with $\Pi$,
    $M^{L_x}(x)=L(x)$ according to the definition of
    access through a consistent language.
\end{proof}

\subsection{Problems $\Avoid$ and $\LOP$}
\begin{definition}[$\Avoid$, Range Avoidance, \cite{KKMP21,korten2022hardest}]\ \\
$\Avoid$ is the following total search problem.\\
\emph{Input:} circuit $C$
with $n$ inputs and $m>n$ outputs.\\
\emph{Output:} $y\in\{0,1\}^{m}\setminus\Img C$.
\end{definition}

\begin{remark}
This problem, in a slightly different formulation, is known in the bounded arithmetic community as the \emph{dual Weak Pigeonhole Principle} (d\problem{WPHP}).
See e.g. \cite{KraGen25} for more details.    
\end{remark}

\noindent Korten \cite{korten2022hardest} have shown that for a stretch of $n+1$
this problem is equivalent to a stretch of $O(n)$
(and, of course, vice versa) under $\P^\NP$ reductions. \\

The following definition is due to Korten and Pitassi \cite{KP24}. It extends the search problem $\MIN$ introduced in \cite{CK98}$,$ to the setting where the input relation is not necessarily a linear order.

\begin{definition}[$\LOP$, Linear Ordering Principle, \cite{KP24}]\ \\
$\LOP$ is the following total search problem.\\
\textbf{Input:} ordering relation $\less$
given as a Boolean circuit $E$
with $2n$ inputs.\\
\textbf{Output:}
either the minimum for $\less$ (that is, $x$ such that 
$\forall y\in\{0,1\}^n\setminus\{x\}\ x\less y$)
or a counterexample, if $\less$ is not a strict linear order.
A counterexample is either a pair satisfying $x\less y\less x$ or a triple satisfying $x\less y\less z\less x$.
\end{definition}

\subsection{Complexity Classes}
The following two definitions have been suggested by Korten and Pitassi
who also proved their equivalence \cite{KP24}.

\begin{definition}[$\LtP$ via reductions]\label{def:LtPred}
A language $L\in\LtP$ if it can be reduced to $\LOP$ using a $\P^{\NP}$-Turing reduction.
(Polynomial-time Turing reductions and polynomial-time many-one reductions have the same effect,
as proved in \cite{KP24}.)
\end{definition}

The following is an alternative definition of $\LtP$, which was shown in \cite{KP24}
to be equivalent.

\begin{definition}[$\LtP$ via symmetric alternation]\label{def:LtPalt}
A language $L\in\LtP$ if there is a ternary relation $R$ computable in time $s(n)$, where $s$ is a polynomial,
\[R\subseteq\{0,1\}^n\times\{0,1\}^{s(n)}\times\{0,1\}^{s(n)},\]
denoted $R_x(u,v)$ for $x\in\{0,1\}^n$, $u,v\in\{0,1\}^{s(n)}$,
such that, for every fixed $x$, it defines a linear order on $s(|x|)$-size strings such that:
\begin{itemize}
    \item for every $x\in L$, the minimal element of this order starts with bit 1,
    \item for every $x\notin L$, the minimal element of this order starts with bit 0.
\end{itemize}
\end{definition}

It is immediate that the latter version of the definition is a particular case
of the definition of $\StP$ \cite{Canetti96,RS98}:
\begin{definition}\label{def:StP}
A language $L\in\StP$ if there is a polynomial-time computable ternary relation 
$R\subseteq\{0,1\}^n\times\{0,1\}^{s(n)}\times\{0,1\}^{s(n)}$,
denoted $R_x(u,v)$ for $x\in\{0,1\}^n$, $u,v\in\{0,1\}^{s(n)}$,
such that:
\begin{itemize}
    \item for every $x\in L$, there exists $w^{(1)}$ such that $\forall v\ R_x(w^{(1)},v) = 1$,
    \item for every $x\notin L$, there exists $w^{(0)}$ such that $\forall u\ R_x(u,w^{(0)}) = 0$.
\end{itemize}
\end{definition}

We now formally define the class $\OtP$ \cite{CR06}, which is the input-oblivious version of $\StP$.
Since we will need also a promise version of it, we start with defining this generalization.

\begin{definition}
\label{def:prOtP}
    A promise problem $\Pi=(\Pi_\YES,\Pi_\NO)$ belongs to\/ $\prOtP$
    if there is a polynomial-time deterministic Turing machine $A$ such that\/
    for every $n\in\mathbb{N}$, there exist $w^{(0)}_n$, $w^{(1)}_n$ (called \emph{irrefutable certificates}) that satisfy for every $x\in\{0,1\}^n$:
    \begin{itemize}
        \item If $x\in\Pi_\YES$, then for every $v$, $A(x,w^{(1)}_n,v)=1$,
        \item If $x\in\Pi_\NO$, then for every $u$, $A(x,u,w^{(0)}_n)=0$.
    \end{itemize}
    No assumption on the behaviour of $A$ is made outside the promise set
    except that it stops (accepts or rejects) in polynomial time.

$\OtP$ is the respective class of languages (that is, it corresponds
to the case of\/ $\Pi_\YES=\overline{\Pi_\NO}$).
\end{definition}

We remind the definition of another oblivious promise class.
\begin{definition}
\label{def:prOMA}
    A promise problem $\Pi=(\Pi_\YES,\Pi_\NO)$ belongs to\/ $\prOMA$
    if there is a polynomial-time deterministic Turing machine $A$
    and,    for every $n\in\mathbb{N}$, there exists $w_n$
    (a witness that serves for every positive instance of length $n$),
    that satisfy
    the following conditions for every $x\in\{0,1\}^n$:
        \begin{itemize}
        \item If $x\in\Pi_\YES$, then 
            $\forall r\ A(x,r,w_n)=1$,
        \item If $x\in\Pi_\NO$, then
            $\forall w\ \Pr_r [A(x,r,w)=1]<1/2$.
        \end{itemize}
\end{definition}        

\medskip
Finally, we define also an input-oblivious version of {\LtP}.

\begin{definition}\label{def:OLtP}
A language $L$ belongs to $\OLtP$ if there is a
polynomial $p$,
polynomial-time deterministic Turing machine $V$
computing a ternary predicate $\{0,1\}^n\times\{0,1\}^{p(n)}\times\{0,1\}^{p(n)}$  (we use the notation $V_x(u,v)$ to denote its result), and two sequences of length-$p(n)$ bit strings
$(y_n)_{n\in\mathbb{N}}$,
$(z_n)_{n\in\mathbb{N}}$
such that
\begin{itemize}
    \item for every $x$, $V_x$ is a strict linear order (define $u<_xv$ iff $V_x(u,v)=1$),
    \item $\forall x\in L\cap\{0,1\}^n\quad1y_n=\min <_x$,
    \item $\forall x\in \overline{L}\cap\{0,1\}^n\quad0z_n=\min <_x$.
\end{itemize}
\end{definition}
\begin{remark}
    One may further require that not only the minimal element but also the entire ordering, for inputs of the same lengths, to coincide. Our results still hold true under this definition as well, although it is unclear whether the two definitions are equivalent or which will ultimately prove more useful.
\end{remark}
\begin{remark}
    Note that unlike $\NP\subseteq\LtP$,
    it is unclear whether $\ONP\subseteq\OLtP$.
    Therefore, when we need both these input-oblivious classes,
    we need to specify both oracles.
\end{remark}

We also define $\SBP$ for the sake of self-completeness.
\begin{definition}[\cite{BGM06}]
\label{def:SBP}
A language $L$ is in $\SBP$ if there exist $\varepsilon > 0$, $k,\ell \in \N$ and a polynomial-time computable predicate $B(x,r)$ such that
\begin{itemize}
\item If $x \in L$ then $\implies \Pr_{r}[B(x,r) = 1] \geq (1 + \varepsilon) \cdot \frac{1}{2^{n^k}} $,  
\item If $x \not \in L$  then $\implies \Pr_{r}[B(x,r) = 1] \leq (1 - \varepsilon) \cdot \frac{1}{2^{n^k}}$. 
\end{itemize}
where $n = \size{x}$ and $r$ is uniformly distributed on $\{0,1\}^{n^\ell}$.
\end{definition}
\section{$\LtP\subseteq\P^{\prSBP}$}
\label{sec:PprSBP}

In this section we prove Theorem \ref{THM:main3}.
Our proof strategy is as follows:
Given a point in a linear order,
we aim to move ``down the order'' (i.e. towards ``smaller'' points).
At each stage we will skip over a constant fraction
of the points remaining on our way to the minimum.
In order to find the next point,
we will employ a binary-search-like procedure
to determine the bits of the desired point, one coordinate at a time. 
Here is where our $\prSBP$ oracle
comes into play:
At each step, we look at the remaining set of points
partitioned into two subsets: the points where the appropriate bit is 0 and where that bit
is 1, and select the subset with the (approximately) smaller
average order rank. 

Before we proceed with the main algorithm,
we show how to approximate the size of a set using a $\prSBP$ oracle.
This procedure could be of independent interest. 

\subsection{Approximate Counting}

In this section we observe that one can approximate deterministically the number of satisfying assignments for a Boolean circuit using a $\prSBP$ oracle (i.e. $\FP^\prSBP$).
Previously, it was shown that this task could be carried out by a randomized algorithm using an $\NP$ oracle ($\class{FBPP}^\NP$) \cite{Stockmeyer85,JVV86} 
and by a deterministic algorithm using a $\prAM$ oracle ($\FP^\prAM$) \cite{Sipser83,GS86}. Note that queries to $\prSBP/\prAM$
can be thought of as queries to specific promise problems.

\begin{definition}[Set-Size Estimation, $\SSE$ and $\WSSE$]
\label{def:SSE/WSSE}
Let $C$ be a Boolean circuit and $m\ge1$ be an integer given in binary representation. 
Then\/ $\SSE \eqdef (\SSE_{\YES}, \SSE_{\NO})$, where
\begin{eqnarray*}
\SSE_{\YES} &=& \condset{(C,m)}{\ct{x}C(x) \geq m} , \\
\SSE_{\NO} &=& \condset{(C,m)}{\ct{x}C(x) \leq m/2}. 
\end{eqnarray*}
If $C$ is a non-deterministic circuit, we denote the corresponding problem by $\WSSE$.
\end{definition}

These two promise problems are complete for promise classes $\prSBP$ and $\prAM$, respectively.
This is what is proved essentially in \cite{BGM06} and \cite{GS86} and formulated explicitly in \cite{Volkovich20}.

\begin{lemma}[Implicit in \cite{BGM06}]
\label{lem:SSE}
$\SSE$ is\/ $\prSBP$-complete.
\end{lemma}

\begin{lemma}[Implicit in \cite{GS86}]
\label{lem:WSSE}
$\WSSE$ is\/ $\prAM$-complete.
\end{lemma}

In particular, these complete problems showcase that $\prSBP \subseteq \prAM$.
The following lemma implies that $\AP \in \FP^{\prSBP}_{\parallel}$ and, in fact, provides a slightly stronger result in the form of one-sided approximation.

\begin{lemma}
\label{lem:SBP-set-approx}
There exists a deterministic algorithm
that given a Boolean circuit $C$ on $n$ variables and
a rational number $\eps > 0$
outputs an integer number $t$ satisfying 
$$\ct{x}C \leq t \leq 4^{\eps/3} \cdot \ct{x}C \leq (1 + \eps)
\ct{x}C $$ 
in time polynomial in $n$, the size of $C$ and $\inveps$,
making non-adaptive oracle queries to $\SSE$. 
\end{lemma}

\noindent The result appears to follow from a combination of previous techniques (and may be considered ``folklore''). For completeness, we now provide its proof. We begin with a definition and a useful inequality. \\

For a unary relation $R(x)$, we denote the number of elements in $R$ as $\ct{x} R(x) \eqdef \size{\condset{x}{x \in R}}$. We will also abbreviate this notation to $\ct{x}R$.
For $k \in \N$, we define $R^{\otimes k}$ -- the $k$-th \emph{tensor power of $R$} as $$R^{\otimes k}((x_1,\ldots,x_k)) \eqdef R(x_1) \wedge R(x_2) \wedge \ldots \wedge R(x_k),$$
where $x_1, \ldots, x_k$ are $k$ disjoint copies of the argument $x$ of $R$, respectively. 

\begin{observation}
\label{obs:tensor}
 Then $\ct{x}(R^{\otimes k}) = ( \ct{x} R)^k$. 
\end{observation}

\begin{observation}[Bernoulli's inequality]
\label{obs:ineq}
Let $0 \leq \eps \leq 1$. Then $4^{\eps/3} \leq 1 + \eps$. 
\end{observation}

\begin{lemma}[Lemma \ref{lem:SBP-set-approx} rephrased]
There exists a deterministic algorithm
that given a Boolean circuit $C$ on $n$ variables and
a rational number $\eps > 0$
outputs an integer number $t$ satisfying 
$$\ct{x}C \leq t \leq 4^{\eps/3} \cdot \ct{x}C \leq (1 + \eps)
\ct{x}C $$ 
in time polynomial in $n$, the size of $C$ and $\inveps$,
making non-adaptive oracle queries to $\SSE$. 
\end{lemma}

\begin{proof}
Let $O$ be a language consistent with $\SSE$.
Set $k \eqdef \ceil{\frac{3}{\eps}}$ and define $\hat{C} \eqdef C^{\otimes k}$. Consider the following algorithm:
 
\begin{enumerate}
   \item If $\left( C, 1 \right) \not \in O$ \ \textbf{return} $t=0$  \Comment{$\ct{x}C = 0$} 
   \item Find $i$ as the smallest $j$ between $1$ and $nk+1$ such that 
   $\left( \hat{C}, 2^j \right) \not \in O$
   \item \textbf{return} $t=\floor{2^{i/k}}$ 
\end{enumerate}

\noindent If $\ct{x}C = 0$ then $\left( C, 1 \right) \not \in O$ and the algorithm outputs $t=0$. Otherwise, $\ct{x}\hat{C} \geq 1$ and hence $\left( \hat{C}, 1 \right) \in O$. On the other hand,
by definition $\ct{x}\hat{C} \leq 2^{nk}$ and hence $\left( \hat{C}, 2^{nk+1} \right) \not \in O$. Therefore, $i$ is well-defined and we have that:

\begin{itemize}
    \item $\left( \hat{C}, 2^i \right) \not \in O \implies \ct{x}\hat{C} < 2^i \implies (\ct{x}C)^k < 2^i \implies \ct{x}C < 2^{i/k}  \implies \ct{x}C \leq t$,
    \item 
    {\renewcommand\implies\Rightarrow
    $\left( \hat{C}, 2^{i-1} \right) \in O \implies \ct{x}\hat{C} > \frac{2^{i-1}}{2} 
    \implies (\ct{x}C)^k > 2^{i-2}
    \implies 4^{1/k} \cdot \ct{x}C > 2^{i/k} \implies
    4^{1/k} \cdot \ct{x}C > t$. }
\end{itemize}
In conclusion, 
$$\ct{x}C \leq t \leq  4^{1/k} \cdot \ct{x}C \leq
4^{\eps/3} \cdot \ct{x}C \leq (1 + \eps) \ct{x}C  .$$

For the running time, the algorithm finds $i$  
by either querying the oracle $O$ at most $nk+2$ times in parallel
or (using binary search) $O(\log_2(nk))$ times sequentially, and thus the computation of $t$ can be carried out in time $\poly(nk)$. 
\end{proof}


\subsection{Estimating the Average Order Rank w.r.t. a Linear Order}
Let $U=\{0,1\}^n$. A single-output
Boolean circuit $E$ with $2n$ inputs induces an ordering relation $\less$ on $U$  as \[x\less y\iff E(x,y)=1.\] If $\less$ is a strict linear order,
we call $E$ a \emph{linear order circuit}.

\begin{observation}
\label{obs:check LOP}
There exists a deterministic Turing machine with $\SAT$ oracle that, given a circuit $E$ on $2n$ variables, 
stops in time polynomial in $n$ and the size of $E$ and does the following:
if $E$ is a linear order circuit, it outputs ``yes'';
otherwise, it outputs a counterexample: a pair satisfying $x\less y\less x$ or a triple satisfying $x\less y\less z\less x$.
\end{observation}

\noindent Fix any strict linear order $<$ on $U$.

\begin{definition}
For an element $\alpha \in U$ we define its order rank as $\ord(\alpha) \eqdef \size{\condset{x \in U}{x < \alpha}}$.
We can extend this definition to non-empty subsets $S \subseteq U$ of $U$ by taking the average order rank:
define $\ord(S) \eqdef \frac{\sum_{x\in S}\ord(x)}{\size{S}}$.
If $S= \condset{x\in U}{C(x)=1}$ is described by a circuit $C$,
we use the same notation: $\ord(C)=\ord(S)$.
\end{definition}

Below are some useful observations that we will use later.

\begin{observation}\label{obs:ord}\hfill%
\begin{itemize}
\item  For a non-empty subset $S \subseteq U$:
$\size{\condset{(\upsilon,\alpha)\in U\times S}{ \upsilon<\alpha}} = \size{S} \cdot \ord(S).$

\item For any $\alpha \in U$: $\ord{\condset{\upsilon \in U}{\upsilon \leq \alpha}} = \ord(\alpha) /2$.

\item Let $S_0, S_1 \subseteq U$ be two non-empty \emph{disjoint}  subsets of $U$. Then \[\ord(S_0 \cup S_1) = \frac{\size{S_0} \cdot \ord(S_0) + \size{S_1} \cdot \ord(S_1)}{\size{S_0} + \size{S_1}}.\]
\end{itemize}
\end{observation}

\begin{proof}These observations follow from the definitions:
\begin{itemize}
\item  $\size{\condset{(\upsilon,\alpha)\in U\times S}{ \upsilon<\alpha}} = \sum \limits_{\alpha \in S} \size{\condset{(\upsilon,\alpha)}{ U \ni \upsilon<\alpha}} =  \sum \limits_{\alpha \in S} \ord(\alpha)
= \size{S} \cdot \ord(S).$
    
\item $\ord{\condset{\upsilon \in U}{\upsilon \leq \alpha}} = \frac{1}{\ord(\alpha)+1} \cdot \sum \limits_{\upsilon \leq \alpha} {\ord(\upsilon )} = \frac{1}{\ord(\alpha)+1} \cdot \frac{\ord(\alpha) \left(\ord(\alpha)+1 \right)}{2} =
\frac{\ord(\alpha)}{2}$.

\item 
$\ord(S_0 \cup S_1) = \frac{1}{\size{S_0} + \size{S_1}} \cdot \left( \sum \limits_{x \in S_0} \ord(x) + \sum \limits_{y \in S_1} \ord(y) \right) = \frac{\size{S_0} \cdot \ord(S_0) + \size{S_1} \cdot \ord(S_1)}{\size{S_0} + \size{S_1}}$.
\end{itemize}
\end{proof}
In the following lemma the $\ord{}$ is defined w.r.t. the order $\less$ described by a linear order circuit $E$.
This lemma allows us to estimate the order rank of a set using a $\prSBP$ oracle.
\begin{lemma}
\label{lem:SBP-ord-approx}
There exists a deterministic algorithm
that given a Boolean circuit $C$ on $n$ variables,
a linear order circuit $E$ on $2n$ variables,
and an $\eps > 0$,  outputs a rational number $r$ satisfying:
$$4^{-\eps} \cdot \ord(C) \leq r \leq 4^{\eps} \cdot \ord(C)$$ 
in time polynomial in $n$, 
the sizes of $C$ and $E$, and in $\inveps$,
given oracle access to $\SSE$. 
\end{lemma}

\begin{proof}
Consider a circuit $D(x,y) \eqdef C(y) \wedge E(x,y)$.
That is, $y$ is accepted by $C$ and $x \less y$.
By Observation \ref{obs:ord}, 
$\ct{(x,y)} D = \ct{x}C \cdot \ord(C)$. 
By Lemma \ref{lem:SBP-set-approx} we can compute integers $t_C$ and $t_D$ that approximate 
the numbers $\ct{x}C$ and $\ct{(x,y)} D$, respectively. 
Formally,
$$\ct{x}C \leq t_C \leq 4^{\eps} \cdot \ct{x}C \text{ and }
\ct{(x,y)}D \leq t_D \leq 4^{\eps}  \cdot \ct{(x,y)}D.$$
Therefore we obtain:
\[ 4^{-\eps}  \cdot \ord(C) \leq \frac{ \ct{(x,y)}D}{4^{\eps}  \cdot \ct{x}C} \leq
\frac{t_D}{t_C} \leq \frac{4^{\eps}  \cdot \ct{(x,y)}D}
{\ct{x}C} = 4^{\eps}  \cdot \ord(C).\qedhere\]
\end{proof}

\subsection{Finding the Minimum Using a $\protect\prSBP$ Oracle}

We use the approximation algorithms developed above in order to
find an element that is much closer to the minimum than
a given element.
The following lemma describes the procedure $\back$ that given an element $\alpha$
finds another element $\beta$ whose order rank is smaller by a constant factor.
We will use this procedure afterwards in order to find
the minimum in a polynomial number of iterations.

The procedure proceeds by determining the bits of the new element,
one coordinate at a time, using a $\prSBP$ oracle. The order rank is w.r.t. the order $\less$
described by a linear order circuit $E$.
\begin{lemma}
\label{lem:SBP-smaller-element}
There exists a deterministic algorithm $\back$
that given a linear order circuit $E$ on $2n$ variables and an element $\alpha \in \bools{n}$, outputs an element $\beta \in \bools{n}$ such that $\ord(\beta) \leq \frac{\ord(\alpha)}{\sqrt{2}}$, 
in time polynomial in $n$ and the size of $E$, given oracle access to $\SSE$. 
\end{lemma}

\begin{proof}
Consider the following procedure:
\medskip

\noindent
$\back(E,\alpha):$
\begin{enumerate}
    \item Define $C(x) \eqdef E(x,\alpha) \vee x= \alpha$.
    \Comment{The set of all elements that are $\less$ than or equal to $\alpha$}
    
    \item Set $\eps = 1/(8n)$. 
    
    \item For $i=1$ to $n$:

\begin{enumerate}
    \item For $b \in \bools{}$: define $C_b \eqdef C \restrict{x_i = b}$ 
    


    \item For $b \in \bools{}$: if $\ct{x}C_b  = 0$ then set
    $C \eqdef C_{1-b},\ \beta_i \eqdef 1-b$; continue to the next $i$

    \Comment{If one of the sets is empty, we choose the other one}

    \item  For $b \in \bools{}$:
    use Lemma \ref{lem:SBP-ord-approx} 
    to approximate $\ord(C_b)$ with $\eps$ into $r_b$
    
    \item If $r_1 \geq r_0$ then $C \eqdef C_0, \beta_i \eqdef 0$ else $C \eqdef C_1, \beta_i \eqdef 1$ 

     \Comment{Choose the set with smaller approximate order}
     \end{enumerate}
\end{enumerate}

\bigskip
After each iteration one more variable $x_i$ gets its value $\beta_i$
and is substituted into $C$, that is, in the current circuit $C$
variables $x_1,\ldots,x_i$ are replaced by the corresponding constants
$\beta_1,\ldots,\beta_i$.
We claim that \emph{after} each iteration the order rank of the resulting circuit
is bounded from the above:
$\ord(C) \leq 4^{2\eps i} \cdot  \frac{\ord(\alpha)}{2}.$ 

Indeed, by Observation \ref{obs:ord}, \emph{before} the first iteration, we have that $\ord(C) =  \frac{\ord(\alpha)}{2}$.  Now consider any iteration. If $C_1$ or $C_0$ are empty, then $\ord(C)$ remains the same and $4^{2\eps i} \leq 4^{2\eps (i+1)}$.
Otherwise, by Lemma \ref{lem:SBP-ord-approx}, for $b \in \bools{}:$ $$4^{-\eps} \cdot \ord(C_b) \leq r_b \leq 4^{\eps} \cdot \ord(C_b).$$ 
If $r_1 \geq r_0$ then $\ord(C_1) \geq r_1 \cdot 4^{-\eps} \geq r_0 \cdot 4^{-\eps} \geq \ord(C_0) \cdot 4^{-2\eps} $ and therefore by Observation~\ref{obs:ord}: 
\begin{multline*}\ord(C) = \frac{\ct{x}C_0 \cdot \ord(C_0) + \ct{x}C_1 \cdot \ord(C_1)}{
\ct{x}C_0 + \ct{x}C_1} \\\geq
\frac{\ct{x}C_0 \cdot \ord(C_0) + \ct{x}C_1 \cdot \ord(C_0) \cdot 4^{-2\eps}}{\ct{x}C_0 + \ct{x}C_1} \geq
\ord(C_0) \cdot 4^{-2\eps}.\end{multline*}
Equivalently, $\ord(C_0) \leq \ord(C) \cdot 4^{2\eps}$.
Similarly, if $r_1 < r_0$ then $\ord(C_1) \leq \ord(C) \cdot 4^{2\eps}.$
Therefore, at each step the order rank is multiplied at most by $4^{2\eps}$.

Consequently, after the $n$-th iteration, $C$ represents the set that contains only the element $\beta$ and we have that
$$ \ord(\beta) 
\leq 4^{2\eps n} \cdot \ord(\alpha) /2 \leq \sqrt{2} \cdot \ord(\alpha) / 2 = \ord(\alpha) / \sqrt{2}.$$

For the runtime, all the steps can be carried out in time polynomial in $n$ and $\inveps = O(n)$.
\end{proof}

Note that the procedure $\back$ has a unique fixed point,
namely, the minimal element.

\begin{observation}
$\back(E,\alpha) = \alpha$ if and only if $\alpha$ is the minimal element in $E$. 
\end{observation}

\begin{proof}
$\ord(\alpha) = 0 \iff \ord(\alpha) \leq \ord(\alpha) / \sqrt{2}.$    
\end{proof}

We are now ready to prove the main result of this section.

\begin{theorem}
\label{thm:LtPinPSBP}
$\LtP\subseteq\P^{\prSBP}$.
\end{theorem}

\begin{proof}
It suffices to construct a deterministic polynomial-time algorithm that solves $\LOP$ given oracle access to $\prSBP$.
Let $E$ be a circuit on $2n$ inputs. We describe an algorithm for $\LOP$.

\begin{enumerate}
    \item Check that $E$ is indeed a linear order using Observation \ref{obs:check LOP}.
    \item Let $\alpha := 0^n,\ \beta := 1^n$.
    \item While $\alpha \neq \beta$ repeat: $\alpha := \beta,\ \beta := \back(E,\alpha)$.
    \item Output $\alpha$.
\end{enumerate}

Given Observation~\ref{obs:check LOP}, we can assume w.l.o.g. that $E$ is a linear order circuit.
We claim that the algorithm will output the minimal element after at most $2n$ iterations. 
Indeed, by Lemma~\ref{lem:SBP-smaller-element}, the order rank of the element $\alpha$ after $2n$ iterations satisfies
$\ord(\alpha) \leq \frac{\ord(1^n)}{\sqrt{2}^{2n}} \leq \frac{2^n-1}{2^n} < 1$,
thus the ``While'' cycle will terminate before that.
\end{proof}

\section{Which Karp--Lipton--style Collapse is Better?}\label{sec:bestKL}

    Chakaravarthy and Roy proved two Karp--Lipton--style collapses:
    down to $\OtP$ \cite{CR06} and down to $\P^\prMA$ \cite{CR11}.
    These two classes seem to be incomparable thereby rising the question:
    which collapse result is stronger?
    We observe that the collapse to $\P^\prMA$ can actually be deepened to $\P^\prOMA$, where $\prOMA$ is the oblivious version of $\prMA$ --- and subsequently show that latter class is contained in both previous classes. That is, $\P^\prOMA \subseteq \P^\prMA \cap \OtP$. 
    Indeed, the ``internal collapse'' of $\prMA$ (and, in fact, even $\prAM$) to $\prOMA$, under the assumption that $\NP\subseteq\Ppoly$, is implicit in \cite{AKSS95}. Nonetheless, we include a formal proof 
    in Section~\ref{subsec:PprOMAcollapse} below 
    (Proposition~\ref{prop:PprOMAcollapse})
    in order to present a self-contained argument:

\begin{center}    If $\NP\subseteq\Ppoly$, then\/ $\prAM \subseteq \prOMA$ and 
            $\PH=\P^\prOMA$.
            \end{center}
    
    To show that this class is not only included in $\P^\prMA$ but also in $\OtP$,
    we use two inclusions:
    \begin{enumerate}
        \item $\P^\prOMA\subseteq\P^\prOtP$.
        \item $\P^\prOtP\subseteq\OtP$.
    \end{enumerate}
    For the first inclusion it suffices to show \[\prOMA\subseteq\prOtP,\]
    which is essentially proved in \cite[Theorem 3]{CR06}: $\MA\subseteq\NOtP$,
    where $\NOtP$ is a less known class that combines $\StP$ with $\OtP$:
    one certificate is input-oblivious, while the other one is not.
    To ``upgrade'' this proof one needs to notice that the proof goes through for promise classes
    with all certificates being input-oblivious.
    As it is not difficult, we include a formal proof of this proposition in Section~\ref{subsec:prOMAinprOtP} below
    (Proposition~\ref{prop:prOMAinprOtP}). However, later 
    (in Section~\ref{subsec:LOtP}) we give a tighter chain of containments
    that uses a different
    method for proving $\P^\prOMA\subseteq \P^\prOtP (=\OtP)$,
    namely, the derandomization technique. It goes through
    newly introduced input-oblivious classes based on symmetric alternation.
    
    The second inclusion \[\P^\prOtP\subseteq\OtP\] seems novel,
    we prove it in Section~\ref{subsec:PprOtPinOtP} (Theorem~\ref{th:PprOtPinOtP}). The containment result
    \[\P^\prOMA\subseteq\OtP\]
    follows (Corollary~\ref{cor:PprOMAinOtP}).

    \subsection{A Karp--Lipton--style Collapse to $\P^\prOMA$}
    \label{subsec:PprOMAcollapse}

    The following proposition shows that the collapse of $\cite{CR11}$
    can be actually pushed down to $\P^\prOMA$.
    For this, we observe that by standard techniques (e.g. \cite{AKSS95}) $\prAM \subseteq \prOMA$ under $\NP\subseteq\Ppoly$.
    \begin{proposition}\label{prop:PprOMAcollapse}
            If $\NP\subseteq\Ppoly$, then\/ $\prAM \subseteq \prOMA$ and\/
            $\PH=\P^\prOMA$.
    \end{proposition}
    \begin{proof}
        If $\NP\subseteq\Ppoly$, then Merlin's proof in $\prMA$
        can be made input-oblivious. Let us show it first for a larger class: $\prAM$.
        If $(\Pi_\YES,\Pi_\NO)\in\prAM$,
        then there is a polynomial-time Turing machine $A$ such that
        \begin{eqnarray*}
            &\forall x\in\Pi_\YES\ &\forall r\ \exists w\ A(x,r,w)=1,\\
            &\forall x\in\Pi_\NO\ &\Pr_r [ \exists w\ A(x,r,w)=1 ]<1/2.
        \end{eqnarray*}
        Let us ask Merlin to send a family $S_n$ of $\CircuitSAT$ circuits of appropriate input sizes
        (slightly abusing the notation: $S_n$ contains circuits for a polynomial range of input sizes,
        not just $n$);
        one can assume that they compute a correct satisfying assignment or say ``no''
        (they can lie only if they say ``no'' for a satisfiable formula).
        Consider the Boolean circuit $A_{x,r}$ obtained by embedding $x$ and $r$
        into $A$, its variables are the bits of the witness $w$.
        Now Arthur can use $S_n$ to produce proofs by himself
        instead of Merlin's original proofs, because the following holds:
        \begin{eqnarray*}
            \forall n\in\mathbb{N}\ \exists S_n\ \forall x\in\Pi_\YES\ &&\forall r\ A(x,r,S_n(A_{x,r}))=1,\\
            \forall x\in\Pi_\NO\ &&\Pr_r[ \exists S_n\ A(x,r,S_n(A_{x,r}))=1 ]<1/2.
        \end{eqnarray*}
        The first condition is true because Merlin can send correct $\CircuitSAT$ circuits.
        The second condition is true because if such a circuit $S_n$ existed, then the original Merlin
        could have sent $w = S_n(A_{x,r})$. Note that in this case although $S_n$ may depend on $x$ and $r$, $A$ will able to catch a cheating Merlin. 
        
        Formally, a new Arthur $A'$ expects $S_n$ as a proof,
        applies $S_n$ to $A_{x,r}$ itself and runs $A$ on the resulting witness.
        Therefore, we get exactly 
        the definition of $\prOMA$ (Def.~\ref{def:prOMA}):
        \begin{eqnarray*}
            \forall n\in\mathbb{N}\ \exists S_n\ \forall x\in\Pi_\YES\ &&\forall r\ A'(x,r,S_n)=1,\\
            \forall x\in\Pi_\NO\ &&\forall S\ \Pr_r [A(x,r,S)=1]<1/2.
        \end{eqnarray*}
        
        Therefore, if $\NP\subseteq\Ppoly$, then $\prAM\subseteq\prOMA$.
        In \cite{CR11}, it is shown that under the same premises $\PH=\P^\prAM$ and hence the claim follows by combining these results.
        \end{proof}

\subsection{Promise Oblivious Merlin--Arthur Protocols are in Promise $\OtP$}
\label{subsec:prOMAinprOtP}   
The following proof essentially repeats the proof of \cite[Theorem 3]{CR06},
which says that $\MA\subseteq\NOtP$,
we need to verify that it holds for promise problems as well,
and Merlin's advice remains oblivious if it was oblivious before,
that is, $\prOMA\subseteq\prOtP$.
    \begin{proposition}\label{prop:prOMAinprOtP}    
            $\prOMA \subseteq \prOtP$.
    \end{proposition}
    \begin{proof}    
Consider a promise problem $\Pi=\Pi_Y\cupdot\Pi_N\in\prOMA$.
There is a polynomial-time deterministic machine $A$ such that
\begin{eqnarray}
            \forall n\in\mathbb{N}\ \exists w_n\ \forall x\in\Pi_\YES\ &&\forall r\ A(x,r,w_n))=1,\label{it:PiY}\\
            \forall x\in\Pi_\NO\ &&\forall w\ \Pr_r [A(x,r,w)=1] <1/2.\label{it:PiN}
\end{eqnarray}
Condition (\ref{it:PiN}) can be replaced by
\[
        \forall n\in\mathbb{N}\ 
        \exists r'_n\ 
        \forall x\in\Pi_N\cap\{0,1\}^n\ 
        \forall w\in\{0,1\}^{p(n)}\ 
        A(x,w,r'_n)=0\tag{\ref{it:PiN}'}
\]
because of the Adleman's trick (similarly to the treatment of $\RP\subseteq\ONP$ in \cite{GoldreichMeir2015} or $\MA\subseteq\NOtP$ in \cite{CR06}): 
given $x$ and $w$,
the new verifier can apply the old one as $A(x,w,r_i)$ for $np(n)$
independent random strings $r_i\in\{0,1\}^{t(n)}$ in order
to reduce the error from $\frac12$ to $\frac12\cdot \frac1{2^{np(n)}}$.
Since for every pair $(x,w)$ there are 
at most a $\frac12\cdot \frac1{2^{np(n)}}$ fraction of strings $r'_n\in\{0,1\}^{tnp(n)}$
results in an error, there exists $r'$ that satisfies (\ref{it:PiN}').
It does not harm the condition (\ref{it:PiY}) as well.
Therefore, Merlin's proof $w$ remains input-oblivious, whereas
Arthur's universal random string $r'_n$ is input-oblivious
and $w$-oblivious.
\end{proof}

\subsection{Merging input-oblivious promise queries}\label{subsec:PprOtPinOtP}

\begin{theorem}[Theorem~\ref{THM:main5}, restated] \label{th:PprOtPinOtP}
 $\P^\prOtP \subseteq \OtP$.   
\end{theorem}

\noindent \textbf{Remark:} Formally, we show that $\P^\Pi \subseteq \OtP$ for every promise problem $\Pi \in \prOtP$.

\begin{proof}
Let $L \in \P^\Pi$ and let $M^\bullet$ be a deterministic oracle machine that decides $L$ correctly given loose oracle access to $\Pi$ (i.e. irrespective of the answers to its queries outside of the promise set), in time $p(n)$  (for a polynomial $p$).
Consider the polynomial-time deterministic verifier $A(q,u,v)$ from the definition of $\Pi\in\prOtP$.
For $n \in \N$, let $1,\ldots,p(n)$ be all possible lengths of oracle queries made by $M$ given an input of length $n$.
Define  
\[W_n \eqdef (w^{(0)}_1,\ldots,w^{(0)}_{p(n)},w^{(1)}_1,\ldots,w^{(1)}_{p(n)})\]
as a vector containing the irrefutable certificates (both ``yes'' and ``no'') of $A$ for the appropriate input lengths. We now construct a new polynomial-time deterministic verifier $A'(x,U,V)$  that will demonstrate that $L\in\OtP$ and
will show that, for any $x$, the string\footnote{There might be different versions of this string as the irrefutable certificates need not to be unique.} $W_{\size{x}}$
constitutes an irrefutable certificate that can be used both as 
$U=(u^{(0)}_1,\ldots,u^{(0)}_{p(n)},u^{(1)}_1,\ldots,u^{(1)}_{p(n)})$ and as
$V=(v^{(0)}_1,\ldots,v^{(0)}_{p(n)},v^{(1)}_1,\ldots,v^{(1)}_{p(n)})$. 

Given $(x,U,V)$ as an input, $A'$ will simulate $M$.
Whenever $M$ makes an oracle query $q$ to $\Pi$,
$A'$ will compute four bits:
\begin{eqnarray*}
&a:=A(q,u^{(1)}_{|q|},v^{(0)}_{|q|}), \qquad
&b:=A(q,v^{(1)}_{|q|},u^{(0)}_{|q|}), \\
&c:=A(q,u^{(1)}_{|q|},u^{(0)}_{|q|}), \qquad
&d:=A(q,v^{(1)}_{|q|},v^{(0)}_{|q|}), 
\end{eqnarray*}
and will proceed with the simulation of $M$ as if the oracle answered $\ell \eqdef (a \land c) \lor (b \land d)$. 

By definition, for any $x$, the machine $M$ computes $L(x)$ correctly given the correct answers to the queries 
in the promise set (i.e. $q \in \Pi_\YES\cup\Pi_\NO$) and irrespective of the oracle's answers outside of the promise set. Thus it suffices to prove that the oracle's answers to the queries in the promise set are computed correctly, which we show now by inspecting the four possible cases.
\begin{itemize}
    \item Suppose $x \in L$ and $U = W_{\size{x}}$. 
    \begin{itemize}
        \item If  $q\in\Pi_\YES$ then 
    $u^{(1)}_{|q|}$ is a `yes'-irrefutable certificate and hence $a=c=1 \implies \ell=1$.
        \item If  $q\in\Pi_\NO$ then 
    $u^{(0)}_{|q|}$ is a `no'-irrefutable certificate and hence $b=c=0 \implies  \ell=0$.
    \end{itemize} 
    
   \item Suppose $x \not \in L$ and $V = W_{\size{x}}$. 
    \begin{itemize}
        \item If  $q\in\Pi_\YES$ then 
    $v^{(1)}_{|q|}$ is a `yes'-irrefutable certificate and hence $b=d=1 \implies  \ell=1$.
        \item If  $q\in\Pi_\NO$ then 
    $v^{(0)}_{|q|}$ is a `no'-irrefutable certificate and hence $a=d=0 \Longrightarrow \ell=0$.\qedhere
    \end{itemize} 
    
\end{itemize}
\end{proof}

\begin{corollary} 
\label{cor:PprOMAinOtP}
     $\P^\prOMA \subseteq \OtP$.
\end{corollary}          
    \begin{proof}
By Proposition \ref{prop:prOMAinprOtP} and Theorem~\ref{th:PprOtPinOtP}.
\end{proof}

\begin{remark}
        When a semantic class without complete problems
is used as an oracle, it may be ambiguous.
However, $\prAM$ does have a complete problem $\WSSE$
(see Definition~\ref{def:SSE/WSSE}).
By inspecting the proof of the collapse (Prop.~\ref{prop:PprOMAcollapse}) 
one can observe that $\WSSE$ actually belongs to $\prOMA$ under $\NP\subseteq\Ppoly$,
and the oracle Turing machine that demonstrates $\PH=\P^\prOMA$
still queries a specific promise problem.

For the inclusion of $\P^\prOMA$ in $\OtP$,
the first part (Prop.~\ref{prop:prOMAinprOtP})
transforms one promise problem into another promise problem,
thus in the inclusion $\P^\prOMA\subseteq\P^\prOtP$
it is also the case that a single oracle is replaced by (another) single oracle.
\end{remark}

\begin{remark}\label{rem:prStPpar}
Unfortunately, the proof of Theorem~\ref{th:PprOtPinOtP}
breaks down when one tries to extend it to $\P^\prStP$
since the queries to $\prStP$ are computed
adaptively and the computation maybe be altered given
the (unpredictable) answers to the queries that lie outside  of promise. 

However, almost the same proof shows that $\P^\prStP_{||}\subseteq\StP$.
Here all the queries can be computed from the input itself, thus instead of the input-oblivious
witness $W_n$ one can give the 
non-input-oblivious witness
specifically for these queries.
The rest of the proof remains intact:
$A'$ will run $A$ on the witnesses
intended for the particular query
and proceed as before.
\end{remark}

\section{$\P^{\prMA}\subseteq \LtP$}\label{sec:PprMA}

In this section we prove Theorem \ref{THM:main1} by, essentially, expanding
the proof of $\MA\subseteq\LtP$ in \cite{KP24}.
We use the following statements from that paper
and an earlier paper by Korten \cite{korten2022hardest}.
We then proceed to the input-oblivious setting.

\subsection{The non-input-oblivious setting}

\begin{definition}[{\protect\cite[Definitions 6,~7]{korten2022hardest}, \protect\cite[Definitions 7,~8]{Korten21}}]
\label{def:PRG}
$\PRG$ is the following search problem: given $1^n$, output a pseudorandom generator\footnote{%
Korten does not say that $m$ has a polynomial dependence on $n$ though we think it is assumed,
and anyway his theorem provides a construction with $m=n^6$.}
$R = (x_1 , \ldots , x_m )$, that is, an array of strings $x_i\in\{0,1\}^n$
such that for every $n$-input circuit $C$ of size $n$:
\[\left|\Pr_{x\gets U(R)}\{C(x)=1\}-\Pr_{y\gets U(\{0,1\}^n)}\{C(y) = 1\}\right| \le \frac1n.\]
\end{definition}
Korten demonstrates that such a generator containing $m=n^6$ strings can be constructed
with a single oracle query to $\Avoid$\footnote{In the paper the $\Avoid$ problem is refered to as ``$\problem{Empty}$''.}
(actually, the stretch is even larger than ``plus 1 output wire'').

\begin{proposition}[{\protect\cite[Theorem 2]{korten2022hardest}, \protect\cite[Theorem 3]{Korten21}}]\label{prop:PRG2Avoid}
{\PRG} reduces in polynomial time to a single {\Avoid} query.
\end{proposition}
Korten and Pitassi demonstrate that $\Avoid$ (which they call $\problem{Weak Avoid}$)
can be solved with one oracle query to $\LOP$.
\begin{proposition}[{\protect\cite[Theorem 1]{KP24}}]\label{prop:Avoid2LOP}
$\Avoid$ is polynomial-time many-one reducible to $\LOP$.
\end{proposition}

The main result of this section is the following theorem.

\begin{theorem}[Theorem \ref{THM:main1}, restated]
\label{thm:PMAinLtP}
    $\P^{\prMA}\subseteq \LtP$.
\end{theorem}
\begin{proof}
    We show how to replace the oracle $\prMA$ by $\LtP$. 
    Since $\P^\LtP=\LtP$ by Definition~\ref{def:LtPred}, the result follows.

    One can assume that calls to the $\prMA$ oracle are made for input lengths such that
    Arthur can be replaced by a circuit $A(x,w,r)$ of size at most $s(n)$ for a specific polynomial $s$.
    One can assume perfect completeness for $A$, 
    that is, for $x$ in the promise set ``$\problem{YES}$'',
    there is $w$ such that $\forall r\ A(x,w,r)=1$.
    
    Before simulating the $\prMA$ oracle, our deterministic polynomial-time Turing machine
    will make oracle calls 
    to $\LtP$ in order to build a pseudorandom generator
    sufficient to derandomize circuits of size $s(n)$.
    By Proposition \ref{prop:PRG2Avoid}, such a pseudorandom generator $G$, which is a sequence $G(1^{s(n)})$ of pseudorandom strings $g_1,\ldots,g_m\in\{0,1\}^{s(n)}$
    for $m$ bounded by a polynomial in $s(n)$, can constructed (for $m=s(n)^6$
    and error $\frac1{s(n)}$) using a reduction to {$\Avoid$}.
    Subsequently, by Proposition \ref{prop:Avoid2LOP},   
    $\Avoid$ is reducible to {$\LOP$}. 
    As a result, $\{g_i\}_{i=1}^m$ can be computed in deterministic polynomial time by querying an {$\LtP$} oracle.

    After $G$ is computed, each call to the $\prMA$ oracle can be replaced
    by an $\NP\subseteq\LtP$ query
    $\exists w\ C(w)$ for the circuit $C$ that computes the conjunction of the circuits $A(x,w,g_i)$ 
    with hardwired $x$ and $g_i$, for every $i$.
    Note that such queries constructed for $x$ outside of the promise set are still valid $\NP$ queries
    even if Arthur does not conform to the definition of $\MA$ in this case.
    These oracle answers are irrelevant, because the original $\P^{\prMA}$ machine must return 
    the correct (in particular, the same) answer irrespectively of the oracle's answer.
\end{proof}

\subsection{The Input-Oblivious Setting}\label{subsec:LOtP}
Korten proves (Prop.~\ref{prop:PRG2Avoid}) that {\PRG} (Def.~\ref{def:PRG}) reduces 
to $\Avoid$, and Korten and Pitassi \cite{KP24} compute $\Avoid$
in $\LtP=\P^\LtP$. Since {\PRG} has a unary input,
one can observe that {\PRG} can be computed using an input-oblivious oracle.
This gives raise to tighter containments.
Indeed, in the non-input oblivious setting 
this containments become equalities.
However, input-oblivious classes lack some of the nice
closure properties and therefore require a special treatment.

\begin{proposition}\label{prop:PRG2OLtP}
{\PRG} can be computed in deterministic polynomial time
with an $\OLtP$ oracle. 
\end{proposition}
\begin{proof} 
    It is shown in \cite{korten2022hardest} (see Prop.~\ref{prop:PRG2Avoid}) that {\PRG}
    consisting of $m=n^6$ strings can be computed in deterministic polynomial time 
    with a single oracle query to $\Avoid$;
    denote the algorithm generating this query by $K$,
    let us assume w.l.o.g. that it outputs a circuit of bit size $n^d$ 
    with $n^c$ inputs,
    for certain integer constants $d>c>2$.
    Since the input to this problem is unary, 
    for each $n$,
    a single instance $C_n$ of $\Avoid$ is to be solved,
    and its solution (i.e. any string outside of the image of $C_n$) solves $\PRG$ (in fact, without further processing).    
    A specific solution to $\Avoid$ can be computed using a deterministic polynomial-time 
    truth-table reduction to a language $L\in\LtP$ \cite{KP24}, 
    that is, there is a polynomial $p$,
    a $p$-time DTM $R$ computing the reduction, 
    a $p$-time DTM $T$ computing the truth table,
    a language $L\in\LtP$ with a polynomial-time verifier $V$ computing an order $<_V$
    such that $T[a_1,a_2,\ldots a_r]$ is a correct solution to $\Avoid$,
    where $r$ is the number of queries made by $R$ (w.l.o.g. they are of the same size),
    and for each $i$, $a_i$ is the first bit of the minimum certificate wrt $<_V$
    applied to the $i$-th query.
    
    Our oracle language $L'\in\OLtP$ is verified by
    the following verifier $U$ (defining its order $<_U$), which
    merges queries using a construction similar to \cite{KP24}. \\

\newcommand\cmt{/$\!\!$/}
\newcommand\lex{\mathrel{<_{lex}}}

For two bit strings $\alpha, \beta$ we define the relation: $\alpha \lex \beta$ to be 1 iff $\alpha$ is lexicographically smaller than $\beta$ (and 0, otherwise). \\
 
\medskip\noindent
    Input: $w$\\
    Certificates: $y$, $z$\\[+3pt]
    Algorithm:
    \begin{enumerate}
        \item If $y=z$ as bit strings, then return $0$. \emph{{\cmt} $y \mathrel{\nless _U} z$}
        \item Compute $t:=|w|$.
        \item Compute $n:=\lfloor \sqrt[d]{|w|}\rfloor$.
        \item Compute $i:=t-w$.

        \item Let $a$ be the first bit of $y$, and $b$ be the first bit of $z$.\\
        \emph{{\cmt} Informally, this bit is a claim for the value of the $i$-th bit of {\PRG}.}
        \item
        If $i>n^c$ then \emph{{\cmt} Out of range.}\\
        \phantom{WWW}if $a\neq b$, then return the result of comparison $a<b$, 
        otherwise return $y \lex z$.

        \item Run $K(1^n)$, denote the resulting circuit by $C_n$.\\
              \emph{{\cmt} Note that the length of its description is $n^d$.}

        \item Run $R(C_n)$ to compute queries $q_1,q_2,\ldots,q_r$.
        
        \item Parse $y$ as $a x y_1\ldots y_{r}$ and $z$ as $b x' z_1\ldots z_{r}$,
              where $a,b\in\{0,1\}$, $x,x'\in \{0,1\}^{n^d}$,\linebreak$y_1,\ldots,y_{r},z_1,\ldots,z_{r}\in\{0,1\}^s$, where $s$ is the bit size of
              elements of $<_V$ for the queries computed by $R$.\\
              \emph{{\cmt} $x,x'$ are candidates for the solution of $\Avoid$ computed by $R,T$.}

        \item Call $y$ syntactically incorrect if either $x[i]$ (the $i$-th bit of $x$) differs from $a$ or $T[y_1[1],\ldots,y_{r}[1]]\neq x$.\\
        \emph{{\cmt} That is, $y$ does not claim that the answer is $x[i]$, or its
        sub-certificates do not yield $x$ as $T$'s answer.}\\[+2mm]
        Define syntactically incorrect $z$ similarly.

        \item If both $y$ and $z$ are syntactically incorrect,
        return $y \lex z$. 

        \item If exactly only one of $y,z$ is syntactically incorrect,
        state that it is greater than the other certificate, return $1$ or $0$ accordingly.

        \item \emph{Now $y,z$ are syntactically correct.}\\
        If $y_1\ldots y_{r}\neq z_1\ldots z_{r}$
        then find the first $j$ such that $y_j\neq z_j$
        and return $V(q_j,y_j,z_j)$.

        \item \emph{We are done and never get to this point:
        the certificates are different,
        but the sub-certificates are equal, 
        thus $x\neq x'$ so one of the certificates 
        had to be recognized as syntactically incorrect by $T$ lookup.}
    \end{enumerate}
    Our verifier $U$ is input-oblivious (it simply ignores the input
    and only uses its length), it computes a linear order,
    and its minimal element starts with the bit
    equal to the $i$-th bit of the pseudorandom generator
    computed by the composition of $K$, $R$, and $T$.
    Thus $U$ defines a language in $\OLtP$,
    and consequent queries $1^{n^d+i}$ to it reveal the bits of
    the pseudorandom generator.
\end{proof}

\begin{corollary}\label{cor:OLtP}\phantom{.}
\begin{enumerate}
    \item 
    $\P^\prMA\subseteq\P^{\OLtP,\NP}\subseteq\LtP$.
    \item 
    $\P^\prOMA\subseteq\P^{\OLtP,\prONP}\subseteq\OtP\cap\LtP$.
\end{enumerate}
\end{corollary}
\begin{proof}
    The first inclusion follows 
    similarly to the proof of Theorem~\ref{thm:PMAinLtP}: 
    the lower DTM can compute $\PRG$ using its $\OLtP$ oracle
    (Prop.~\ref{prop:PRG2OLtP})
    and then use it to derandomize $\prMA$ (resp., $\prOMA$) using its $\NP$ (resp., $\prONP$) oracle. 
    
    The second inclusion in all the items follows from $\NP\subseteq\LtP$, $\prONP\subseteq\prOtP$, 
    $\OLtP\subseteq\LtP$ (syntactically), $\OLtP\subseteq\OtP$ (similarly to $\LtP\subseteq\StP$),
    $\P^\LtP=\LtP$ \cite{KP24}($=\P^\prLtP$, because $\LtP$ is a syntactic class), $\P^{\OtP}=\OtP$ ($=\P^\prOtP$ (Theorem~\ref{th:PprOtPinOtP})).
\end{proof}

\subsection{An even better Karp--Liption--style collapse?}\label{subsec:OprMA}
One can define a promise version of $\prOMA$ where
the machine satisfies the bounded-error promise
on the promise set but it has to satisfy
the input-oblivious promise everywhere.
\begin{definition}
\label{def:OprMA}
    A promise problem $\Pi=(\Pi_\YES,\Pi_\NO)$ belongs to\/ $\OprMA$
    if there is a polynomial-time deterministic Turing machine $A$
    and,    for every $n\in\mathbb{N}$, there exists $w_n$
    (a witness that serves for every positive instance of length $n$),
    that satisfy
    the following conditions:   
        \begin{itemize}
        \item If $x\in\Pi_\YES\cap\{0,1\}^n$, then 
            $\forall r\ A(x,r,w_n)=1$,
        \item If $x\in\Pi_\NO$, then
            $\forall w\ \Pr_r [A(x,r,w)=1]<1/2$,
        \item If $x\notin\Pi_\YES$, then
            $\forall w\ \Pr_r [A(x,r,w)=1]<1$.
        \end{itemize}
\end{definition}
We conjecture that by careful inspection of the proof of \cite{CR11}
one can check that the only reason to use promise problems is a lack
of guarantee for bounded probability of error, therefore:
\begin{conjecture}
    The Karp--Lipton--style collapse of Proposition~\ref{prop:PprOMAcollapse}
    to $\P^\prOMA$ is actually to $\P^\OprMA$.
\end{conjecture}

Anyway, this definition gives rise to an even nicer
corollary than Corollary~\ref{cor:OLtP}.
\begin{corollary}\label{cor:OprMA}\phantom{.}
    $\P^\OprMA\subseteq\P^{\OLtP,\ONP}\subseteq\OtP\cap\LtP$.
\end{corollary}
\begin{proof}
    The second inclusion is already proved in Corollary~\ref{cor:OLtP},
    we need to show only the first one. It's also analogous:
    similarly to the proof of Theorem~\ref{thm:PMAinLtP}: 
    the lower DTM can compute $\PRG$ using its $\OLtP$ oracle
    (Prop.~\ref{prop:PRG2OLtP})
    and then use it to derandomize $\OprMA$ using its $\ONP$ oracle
    (we do not need a $\prONP$ oracle here, because the promise
    in $\OprMA$ does not concern the input-obliviousness).
\end{proof}

\section{Discussion and Further Research}\label{sec:OQ}

\paragraph{Some Unresolved Containments}
\begin{enumerate}

\item
Can one strengthen our inclusion $\LtP\subseteq \P^\prSBP$
to $\StP\subseteq\P^\prSBP$? 
One can try combining our techniques with the proof
of $\StP\subseteq\P^{\prAM}$ by Chakaravarthy and Roy \cite{CR11}.

Note that in the other direction
it is open even whether $\SBP \subseteq \StP$.

\item Chakaravarthy and Roy \cite{CR11} asked
whether $\P^\prMA$ and $\P^\prStP$ are contained in $\StP$.
While we resolved the first question, the second one remains open
We note that, although in the input-oblivious world both inclusions hold
($\P^\prOMA\subseteq\P^\prOtP\subseteq \OtP$, Corollary~\ref{cor:PprOMAinOtP}),
the proof of the latter inclusion (Theorem~\ref{THM:main5}) is essentially input-oblivious
(one needs to give all the certificates for the oracle non-adaptively,
and queries cannot be predicted because oracle answers cannot be predicted
for promise problems, thus in the non-input-oblivious setting it works for parallel queries only (Remark~\ref{rem:prStPpar})).

\item As was mentioned, the $\FP^{\prSBP}$ procedure for approximate counting can be implemented in $\FP^{\prSBP}_\parallel$ --- that is, using parallel (i.e. non-adaptive) oracle queries. On the other hand the containment $\LtP \subseteq \P^\prSBP$, which uses approximate counting as a black-box subroutine, seems to require sequential, adaptive queries. Could one implement the latter containment using parallel queries 
(i.e. show that  $\LtP \subseteq \P^\prSBP_\parallel$)?
In particular, as $\class{PP}$ is consistent with $\prSBP$ \cite{BGM06} and is closed under non-adaptive Turing reductions \cite{FR96}, this would imply that $\LtP \subseteq \class{PP}$.
Note that it is unknown even whether $\P^{\NP}\subseteq\class{PP}$, while
$\P^{\NP}\subseteq \LtP$. Moreover, there is an oracle separating the former two classes \cite{Beigel94}.

\item  A recent work of Gajulapalli et al. \cite{GGLS25} places 
$\LtP$ in the class\footnote{$\class{UEOPL}$ consists of problems that are many-one polynomial-time reducible to $\problem{Unique-End-of-Potential-Line}$, see \cite{FGMS20,GGLS25}.} $\class{UEOPL}^\NP$, which appears to be incomparable to our result (Theorem \ref{THM:main3}). Can one determine the relative status of  $\class{UEOPL}^\NP$ and $\P^\prSBP$?

\item  Similarly to $\LtP$, one could define a class of languages reducible to $\Avoid$.
A similar class of search problems, $\class{APEPP}$, has been defined by \cite{KKMP21,korten2022hardest}
(and Korten \cite{korten2022hardest} proved that constructing a hard truth table
is a problem that is complete for this class under $\P^\NP$-reductions);
however, we are asking about a class of languages.
Korten and Pitassi have shown that $\LtP$ can be equivalently defined
using many-one, Turing, or $\P^\NP$-reductions,
thus there are several options.
One can observe that the containment $\P^\prMA\subseteq \LtP$ (Theorem~\ref{thm:PMAinLtP})
is essentially proved via the intermediate class $\P^{\Avoid,\NP}$
that uses both an oracle for \emph{Range Avoidance}
(a single-valued or an essentially unique \cite{KP24} version) and an oracle for $\SAT$.
Can one prove that one of the containments in
$\P^\prMA\subseteq \P^{\Avoid,\NP} \subseteq \LtP$
is in fact an equality?

\item We defined an input-oblivious version $\OLtP$ of $\LtP$,
it is not obvious whether $\ONP\subseteq\OLtP$ or even $\ONP\subseteq\P^\prOLtP$.

\end{enumerate}

\paragraph{Relations between collapse results and ``hard'' functions.}
Starting from Karp--Lipton's paper \cite{KarpLipton80}, Kannan's fixed-polynomial circuit complexity lower bounds \cite{Kannan82} were improving accordingly to new collapses: if a new collapse $\NP\subseteq\Ppoly \implies \PH=\mathcal{C}$ is shown for a class $\mathcal{C}$ containing $\NP$,
it immediately implies lower bounds for this class, because if $\NP\not\subseteq\Ppoly$, we are already done.

However, a collapse to $\OtP$ \cite{CR06} did not imply lower bounds for $\OtP$, because $\NP$ is unlikely to be contained in it (after all, $\OtP\subseteq\Ppoly$). It was not until nearly two decades later that lower bounds for $\OtP$ have been shown by Gajulapalli, Li, and Volkovich \cite{GLV24} building on recent progress for the range avoidance problem \cite{korten2022hardest,li2023symmetric}, thus matching the progress on the two questions again.

Korten and Pitassi \cite{KP24} have shown fixed-polynomial lower bounds for their new class $\LtP$ without showing a collapse result, thereby introducing a misalignment once again, yet this time in the opposite direction. Our paper's inclusion $\P^\prMA\subseteq\LtP$ restores the balance.

However, the observation that in the input-oblivious world the currently best collapse $\P^\prOMA$ reopens this question. Does this class possess fixed-polynomial circuit lower bounds?
One can 
observe that Santhanam's proof \cite{Santhanam09}
of $\bsize[n^k]$ lower bounds for promise problems in $\prMA$
is input-oblivious. Indeed, the presented hard promise problems are actually in $\prOMA$! 
However, these \textbf{promise problems} do not yield a \textbf{language} in $\P^\prOMA$ that is hard for $\bsize[n^k]$,
and we leave this question for further research.

The best class for which fixed-polynomial circuit lower bounds can be proved (trivially) is $\P^{\ehardtt}$ (for any particular fixed $\eps > 0$), where $\ehardtt$,  asks\footnote{Korten \cite{korten2022hardest} defines a smoother version of it where the input length is not necessarily a power of two.} given $1^{2^n}$,
to output a truth table of a function $\{0,1\}^n\to\{0,1\}$ of circuit complexity at least $2^{\varepsilon n}$.
Can one prove a collapse to this class (or at least to $\P^\Avoid$)? Note that for the purpose of fixed-polynomial lower bounds
even a limited version of $\ehardtt$ suffices where the truth table is non-empty for a
number of entries greater than any polynomial and its complexity is only superpolynomial.


\section*{Acknowledgment}
The authors are grateful to Yaroslav Alekseev for discussions,
to Dmitry Itsykson for discussing and proofreading 
a preliminary version of this paper,
and to Jan Kraj\'\i\v{c}ek for providing
multiple references and for useful discussions.

\newcommand{\etalchar}[1]{$^{#1}$}

\end{document}